\documentclass[envcountsame,11pt]{llncs}
\usepackage{geometry}
\usepackage{multirow}
\usepackage{tikz}
\usepackage{subfigure}
\usepackage{inputenc}
\usepackage{amssymb}
\usepackage{amsmath}
\usepackage{array}
\usepackage{listing}  
\usepackage{underscore}
\usepackage{graphicx}
\usepackage{eso-pic}
\usepackage{fancyhdr}
\usepackage{pstricks}
\usepackage{algorithmic}
\usepackage{algorithm}
\usepackage{hyperref}
\usepackage[all]{hypcap}
\usepackage{float}
\usepackage{pdfpages}
\usepackage[inline]{enumitem}
 \usepackage{hyperref}

\usepackage{algorithm,algorithmic,verbatim}
\usetikzlibrary{shapes,snakes,positioning,calc}
\usetikzlibrary{arrows,shadows}
\usepackage{hyperref}
\usepackage[
	lambda,
	operators,
	advantage, 
	sets,     
	adversary,
	landau,
	probability,
	notions,	
	logic,
	ff,
	mm, 
	primitives,
	events,
	complexity,
	asymptotics,
	keys]{cryptocode}

\usepackage{csquotes}

\usepackage{dashbox}
\usepackage{todonotes}

\usepackage{url}

\usetikzlibrary{shapes.callouts}

\usepackage{listings}

\usepackage{trace}

\usepackage{makeidx}
\usepackage{mdframed}

\makeindex

\definecolor{mygreen}{rgb}{0,0.6,0}
\definecolor{mygray}{rgb}{0.1,0.1,0.1}
\definecolor{mymauve}{rgb}{0.58,0,0.82}

\newtheorem{Claim}[theorem]{Claim}

 \newtheorem{observation}[theorem]{Observation}

\renewcommand{\bool}{\{0,1\}}
\renewcommand{\booln}{\{0,1\}^n}

\newcommand{\etal}{{\it et al }}   
\newcommand{\eg}{{\sl e.g.~}}

\newcommand{\eq}{=}
\newcommand{\pl}{+}

\newcommand{\ex}{\mathsf{Ex}}

\renewcommand{\negl}{\textsf{negl}}

\newcommand{\eqdef}{\stackrel{\text{def}}{=}}

\title{Crooked Indifferentiability Revisited}

\author{Rishiraj Bhattacharyya\inst{1} \and Mridul Nandi\inst{2} \and Anik Raychaudhuri\inst{2}}
 \institute{ NISER, HBNI, India, {\tt rishiraj.bhattacharyya@gmail.com} \and Indian Statistical Institute, Kolkata, India, {\tt mridul.nandi@gmail.com,anikrc1@gmail.com}}

\begin{document}
\maketitle             
 \begin{abstract}
   
In CRYPTO 2018, Russell \etal introduced the notion of crooked indifferentiability to analyze the security of a hash function when the underlying primitive is subverted. They showed that the $n$-bit to $n$-bit function implemented using enveloped XOR construction (\textsf{EXor}) with $3n+1$ many $n$-bit functions and $3n^2$-bit random initial vectors (iv) can be proven secure asymptotically in the crooked indifferentiability setting.

\begin{itemize}
\item We identify several major issues and gaps in the proof by Russel \etal, We show that their proof can achieve security only when the adversary is restricted to make queries related to a single message.   
\item We formalize new technique to prove crooked indifferentiability without such restrictions. Our technique can handle function dependent subversion. We apply our technique to provide a revised proof for the \textsf{EXor} construction.
\item We analyze crooked indifferentiability of the classical sponge construction. We show, using a simple proof idea, the sponge construction is a crooked-indifferentiable hash function using only $n$-bit random iv. This is a quadratic improvement over the {\sf EXor} construction and solves the main open problem of Russel \etal \end{itemize}

 \end{abstract}
%

\section{Introduction}
\label{sec:introduction}

\noindent\textsc{Blackbox Reduction and Kleptographic attack.} Many of the modern cryptographic constructions are analyzed in a modular and inherently black-box manner. The schemes or protocols are built on underlying primitives only exploiting the functionality of the primitives. While analyzing the security, one shows a reduction saying, a  successful attack on the construction will lead to an attack against the underlying primitive. Unfortunately, this approach completely leaves out the implementation aspects. While the underlying primitive may be well studied, a malicious implementation may embed trapdoor or other sensitive information that can be used for the attack. Moreover, such implementation may well be indistinguishable from a faithful implementation.  These type attacks fall in the realm of \emph{Kleptography}, introduced by Young and Young\cite{C:YouYun96,EC:YouYun97}. The real possibility of Kleptographic attacks has been confirmed by Snowden's revelation. Recently, starting with the work of Bellare, Paterson and Rogaway \cite{C:BelPatRog14}, the attention to this setting has been rejuvenated. BPR showed that it is possible to mount an algorithm-substitution-attack against almost all known symmetric key encryption scheme to the extent that the attacker learns the secret key.  A series of work has been done in the recent years formalizing approaches resisting algorithm-subversion attacks \cite{EC:DGGJR15,EC:BelHoa15,EC:MirSte15,FSE:DegFarPoe15,C:DPSW16,AC:RTYZ16,CCS:RTYZ17,ACNS:AFMV19,CCS:AteMagVen15}. In his IACR Distinguished lecture, Rogaway urged for community-wide efforts to work on defending against Kleptography. 
%
%
%

\noindent\textsc{Crooked Indifferentiability.} In CRYPTO 2018, Russel \etal\cite{C:RTYZ18} introduced the notion of crooked indifferentiability as a security notion for hash functions in the kleptographic setting. Indeed, hash functions are ubiquitous tool of modern cryptology, and natural choice for instantiating Random Oracles in practical protocols. Subverting internal algorithms of a hash function will be a natural target for kleptographic attack. In the crooked-indifferentiability setting, thee constructions are randomized. The distinguisher can \emph{substitute} the ideal primitive $f$ by a subverted implementation which ``crooks'' the function on $\epsilon$ fraction of the inputs, producing a crooked primitive $\tilde{f}$. In order to prove crooked-indifferentiability of a construction $C^f$, one constructs a simulator $S$ such that $(C^{\tilde{f}}(.,R),f)$ and $(\mathcal{F},S^{\tilde{f}}(R))$ are indistinguishable. As in the case of classical indifferentiability \cite{TCC:MauRenHol04,C:CDMP05}, the notion of crooked indifferentiability and the corresponding composition theorem aims to formalize the idea of instantiating hash functions as Random Oracles in protocols, and prove the kleptographic security.\smallskip

Unfortunately, not much theory is known about crooked indifferentiability except the results by \cite{C:RTYZ18}. Specifically, it is unknown whether the popular modes of operations used in cryptographic hash functions can be proven secure under this notion. Looking at the potential application in the analysis of hash functions in the post-Snowden era, we believe an in-depth study of the crooked-indifferentiability notion is an extremely important problem of theoretical cryptography.
\begin{figure}[!htb]
  \begin{minipage}{0.5\linewidth}
  \centering
  \includegraphics[scale=0.6]{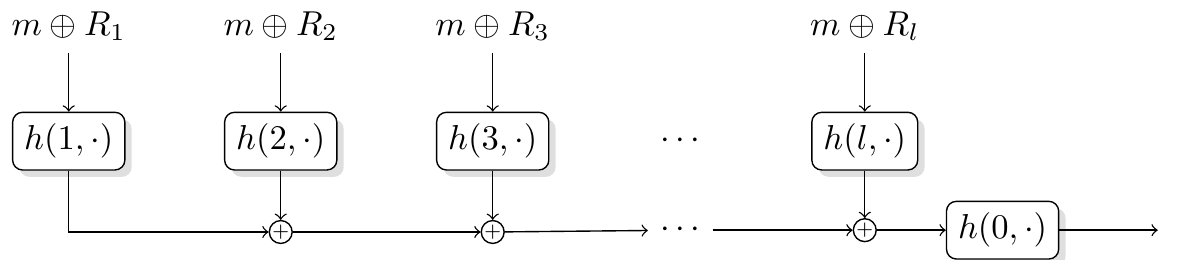}
 
  \label{fig:exor-pic}
\end{minipage}  
  \begin{minipage}{0.5\linewidth}
      \includegraphics[scale=0.5]{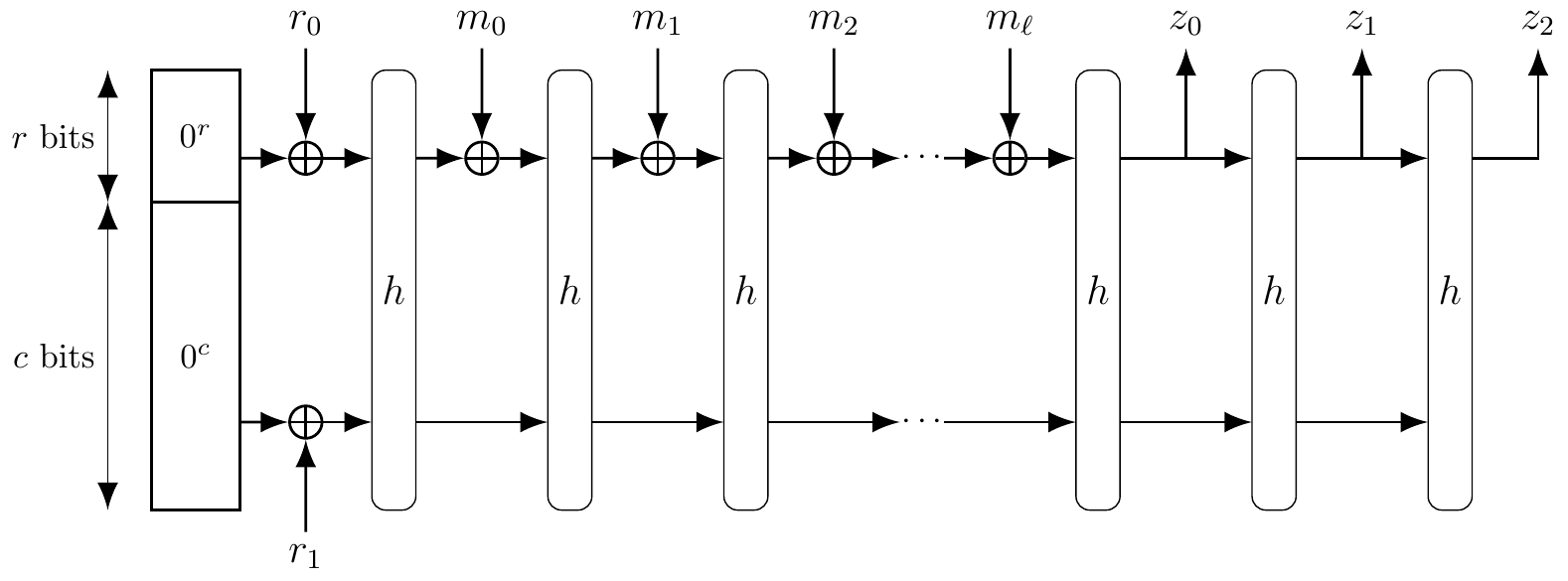}
  \label{fig:spongepic}
\end{minipage}

  \caption{Sponge and {\sf EXor} constructions.}
  \label{fig:results}
\end{figure}
\subsection{Our Contribution}
\label{sec:our-contribution}
In this paper, we introduce new techniques to prove crooked-indifferentiability of domain extension techniques. We proved crooked indifferentiability of {\sf EXor} construction (Figure \ref{fig:results} left) and the sponge construction (Figure \ref{fig:results} right).

\noindent\textbf{Another Look at RTYZ proof.}
We uncover that the proof technique of RTYZ \cite{C:RTYZ18} only works under simplifying conditions. The crooked-indifferentiability proofs of {\sf EXor} do not go through in the general setting. The simplest-to-explain issue is the following. RTYZ silently assumes that it is sufficient to prove indifferentiability showing consistency of the simulator for a single message. In other words, they implicitly restrict the indifferentiability distinguisher to make at most a single query to the construction (resp. RO in the ideal world). When applied with the composition theorem, this restriction limits the \emph{environment}, which includes the protocols and the adversary, to make in total a single query to the Random Oracle (RO). Thus, {\bf the result of \cite{C:RTYZ18} is applicable only if the protocol queries the RO once and the adversary does not even get access to the RO}. This makes the results of \cite{C:RTYZ18} invalid for any reasonable setting and the results built on it, \eg \cite{PKC:CRTYZZ19}, gets nullified. As we show in Section \ref{sec:critic}, the technique collapses if the distinguisher can check consistency for more than one message. Thus the results of \cite{C:RTYZ18}, while involves heavy technical machineries, fail to achieve any meaningful security guarantee.  In Section \ref{sec:critic}, we discuss this and other technical issues in detail. \smallskip

\noindent\textbf{Enveloped Xor Construction and its crooked-indifferentiability.}
In \cite{C:RTYZ18}, Russel \etal proved crooked-indifferentiability of the enveloped-xor construction, albeit as we just discussed, in a very restricted setting. We prove security without any such restriction. Our techniques, in comparison, are also quite simple and use basic tools of probability. \smallskip

\noindent\textbf{Crooked-indifferentiability of Sponge construction.}
While enveloped-xor construction serves as a proof of concept of crooked-indifferentiability, the construction is far from being practical. It uses $3n+1$ \emph{independent} hash functions (one call to each) and $3n^2$ bits of randomness in order to construct an $n$-bit to $n$-bit subversion resilient function! When we instantiate that to construct practical hash functions, then this cost will amplify and we end up with somewhat impractical construction. Constructing a crooked-indifferentiable hash function that uses $\mathcal{O}(n)$ random bits, and makes linear number of calls to the underlying primitive, was left open as the central question by \cite{C:RTYZ18}.

We prove crooked-indifferentiability of sponge construction, the construction underlying the SHA3 hash function, with random initialization vector. We assume the underlying primitive to be a random function. The construction uses only $n$-bit randomness and even with the most conservative parameter choices, makes at most $2n$ many calls to the underlying primitive to construct an $n$-bit to $n$-bit function. The result is not limited to produce length preserving hash functions. In fact, similar to the security enjoyed by the sponge construction in the classical setting, our result is applicable for the general sponge construction hashing $\ell$-bit messages to $s$-bit digests.




\smallskip

\section{Notations and Preliminaries}
\label{sec:notat-prel}
\noindent\textsc{Notations}. For any tuples of pairs $\tau \eq ( (x_1, y_1), \ldots, (x_s, y_s))$ we write $\mathcal{D}(\tau)$ (called domain of $\tau$) to denote the set $\{x_i : 1 \leq i \leq s \}$. $\tau_j$ denotes the sub-list containing the first $j$ entries in $\tau$. $\tau_j \eq ( (x_1, y_1), \ldots, (x_j, y_j))$.  We write $x\sample S$ to denote the process of choosing $x$ uniformly at random from a set $S$ and independently from all other random variables defined so far. For a positive integer $l$, we use $(l]$ to denote the set $\{1,\cdots,l\}$. $[l]$ is used to denote the set $\{0, 1,\ldots, l\}$. \smallskip

\noindent\textsc{Class of Functions}. 
The positive integer $n$ is our security parameter. Let $\mathcal{D} :\eq [l] \times \bin^n$. Let $\mathsf{H}$ denote the set of all functions from $\mathcal{D}$ to $\bin^n$. Similarly, $\mathsf{F}$ denote the set of all functions from $\bin^n$ to $\bin^n$. For any $z :\eq ((a_1, b_1), \ldots, (a_{q_1}, b_{q_1}))$ with $b_1, \ldots, b_{q_1} \in \bin^n$ and distinct $a_1, \ldots, a_{q_1} \in \mathcal{D}$, we write $h \vdash z$ if $h(a_i) \eq b_i$ for all $i \in (q_1]$. We denote the set of all functions $h$ such that $h \vdash z$ as $\mathsf{H}|_z$. \smallskip

\noindent\textsc{Distinguishing Advantage}. In this paper, we measure the efficiency of algorithms by their query complexity.  An oracle algorithm $\adv$ having access of one oracle is called $q$-query algorithm if it makes at most $q$ queries to its oracle. Similarly, an oracle algorithm having access to two oracles is called $(q_1, q_2)$-query algorithm, if it makes at most $q_1$ and $q_2$ queries to its first and second oracles respectively. We use $X^t$ to denote the $t$-tuple $(X_1, \ldots, X_t)$. 

\begin{definition}[Distinguishing Advantage]
\label{def:advantange}
Let $F^l $ and $G^l$ be two $l$-tuple of probabilistic oracle algorithms for some positive integer $l$. We define \emph{advantage} of an adversary $\adv$ at distinguishing $F^l $ from $G^l $ as
$$ \Delta_\adv(F^l\ ;\ G^l) = \left| \Pr[\adv^{F_1,F_2,\cdots, F_l} = 1] - \Pr[\adv^{G_1,G_2,\cdots, G_l} = 1] \right|.$$
\end{definition}

\section{Modeling Subversion Algorithms and Crooked Indifferentiability}
\label{sec:crooked}
We recall the related terms and notations introduced in \cite{C:RTYZ18} in our terminologies. \vskip3pt
\noindent\textsc{Implementor}. A $(q,\tau)$ implementor is a $ q$ query oracle algorithm $\adv^{\mathcal{O}}$ which outputs a $\tau$ query oracle algorithm $\tilde{H}^{\mathcal{O}}$, called the implementation. Let $z$ denote the transcript of oracle queries of $\adv$. For the rest of the paper, we shall assume that $\mathcal{O}$ is a function $h$, sampled uniformly at random from $\mathsf{H}$. The implementation $\tilde{H}$ is correct if for all $h \in \mathsf{H}$ and for all $x \in \mathcal{D}$,  $\tilde{h}(x) \eqdef \tilde{H}^h(x) \eq h(x)$. As $h$ is a random function, we assume without loss of generality that $\tilde{H}^h(x)$ queries $h(x)$ and it is in fact the first query it makes. The transcript $z$ is hardwired in $\tilde{H}$ and all  the $\tau$ queries made by $\tilde{H}$ are different from $\mathcal{D}(z)$. \smallskip


	


\noindent Let $\mathcal{Q}^h(\alpha)$ denote the set of all queries during the computation of $\tilde{h}(\alpha)$. 
We use $\alpha\twoheadrightarrow_h\alpha'$ to denote that $\tilde{h}(\alpha)$ queries $h(\alpha')$. Similarly, $\alpha\not\twoheadrightarrow_h\alpha'$, denotes that $\tilde{h}(\alpha)$ does not query $h(\alpha')$.  $\tilde{\alpha}_j$ denotes the $j^{th}$ query made by $\tilde{h}(\alpha)$.

\begin{definition}[crooked implementor]
	A $(q, \tau)$ implementor $\adv_1$ is called {\em $\epsilon$-crooked} for $\mathsf{H}$, if for every $h \in \textsf{H}$ and for all $0 \leq i \leq l$, $\Pr_{\alpha\sample \bin^n}(\tilde{h}(i, \alpha) \neq h(i, \alpha)) \leq \epsilon$.   
\end{definition}
\textbf{Summary.} A (crooked) implementation $\tilde{h}$, to compute $\tilde{h}(\alpha)$, queries $h(\alpha_1),\cdots, h(\alpha_\tau)$ on $\tau$ many distinct points ($\alpha_1=\alpha$) and its decision of whether to subvert $h(\alpha)$ depends on this transcript and the hardwired string $z$. For an $\epsilon$-crooked implementation, for each $h\in\mathsf{H}$, for at most $\epsilon$ fraction of $\alpha\in\mathcal{D}$, $h(\alpha)$ is subverted.

\noindent\textsc{Crooked Distinguisher.} A crooked distinguisher is a two-stage adversary; the first stage is a subverted implementor and the second stage is a distinguisher. 

\begin{definition}[crooked distinguisher]
We say that a pair $\adv :\eq (\adv_1, \adv_2)$ of probabilistic algorithms {\em $((q_1, \tau,\epsilon), q_2)$-crooked distinguisher } for $\mathsf{H}$ if 

(i) $\adv_1$ is a $\epsilon$-crooked $(q_1, \tau)$ implementor for $\textsf{H}$ and 

(ii) $\adv_2(r,z, \cdot)$ is a $q_2$-query distinguisher where $r$ is the random coin of $\adv_1$, and $z$ is the transcript of interaction of $\adv_1$ with $h$. 

\end{definition}

\noindent\textsc{Crooked Indifferentiability.} Now we state $\textsf{H}$-crooked indifferentiable security definition (as introduced in \cite{C:RTYZ18}) in our notation and terminology.

\begin{definition}[$\mathsf{H}$-crooked indifferentiability \cite{C:RTYZ18}]\label{def:crooked-old}
	Let $\mathcal{F}$ be an ideal primitive and $C$ be an initial value based $\mathcal{F}$-compatible oracle construction. The construction $C$ is said to be {\bf $((q_1, \tau)$, $ (q_2, q_{\mathrm{sim}})$, $ \epsilon, \delta)$-$\mathsf{H}$-crooked indifferentiable from $\mathcal{F}$} if there is a $q_{\mathrm{sim}}$-query algorithm $S$ (called simulator) such that for all $((\epsilon, q_1, \tau), q_2)$-crooked distinguisher $(\adv_1(r)$, $\adv_2(r, \cdot, \cdot))$ for $\mathsf{H}$, we have 
	\begin{equation}
	\Delta_{\adv_2(r,z, R)}\big((h, C^{ \tilde{h} }(R, \cdot))\ ;\ (S^{\mathcal{F},\tilde{h}}(z, R), \mathcal{F}) \big) \leq \delta
	\end{equation} 
	where $z$ is the advise string of $\adv_1^h$ and $R$ is the random initial value of the construction sampled after subverted implementation is set.  
\end{definition}
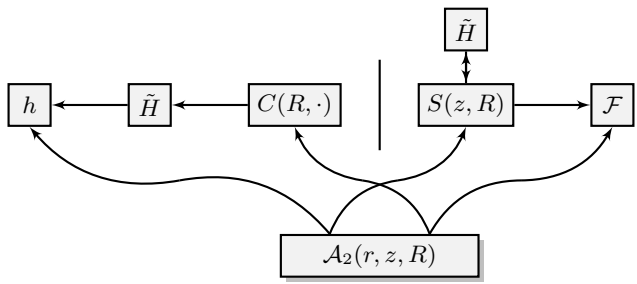
\begin{figure}[h]
	\begin{center}
		\begin{tikzpicture}[auto, node distance=1cm, >=latex',scale=0.8]
\begin{scope}
\tikzstyle{Attacker} = [draw, fill=black!05, rectangle, 
    minimum height=1.7em, minimum width=8em, drop shadow,thick]
\tikzstyle{Box} = [draw, fill=black!05, rectangle, 
    minimum height=1.7em, minimum width= 1.7em, thick]
\tikzstyle{to} = [->,thick]
\tikzstyle{line}= [-,thick]
\tikzstyle{dotto} = [->,dotted, thick]

\node [Box, name=C] {$C(R, 
	\cdot)$};
\node [Box, left =of C, name=F1] {$\tilde{H}$};
\node[Box, left=of F1, name =F]{$h$};
\node [Box, right=of C, name=S] {$S(z, R)$}; 
\node [Box, right=of S, name=G] {$\mathcal{F}$};
\node[Box, above of =S, name=h]{$\tilde{H}$};
\draw [draw,to] (C) -- (F1);
\draw[draw, to](F1)-- (F);
\draw [draw, to] (S) -- (G);
\draw [draw, to] (S) -- (h);
\draw [draw, to] (h) -- (S);

\coordinate (mid) at ($0.5*(C.east) + 0.5*(S.west)$);

\draw [draw, line] (mid) -- ++(0,0.75);
\draw [draw, line] (mid) -- ++(0,-0.75);

\node [Attacker, below of=mid, node distance=2cm, name=D] {$\mathcal{A}_2(r,z,R)$};

\coordinate (dleftmid) at ($0.5*(D.north west) + 0.5*(D.north)$);
\coordinate (drightmid) at ($0.5*(D.north east) + 0.5*(D.north)$);

\draw [thick] (dleftmid) to[bend right] ($0.5*(F.south)+0.5*(dleftmid)$);
\draw [to] ($0.5*(F.south)+0.5*(dleftmid)$) to[->, bend left, thick] (F.south);
\draw [thick](dleftmid) to[bend left, thick, dotted] ($0.5*(S.south)+0.5*(dleftmid)$);
\draw [to] ($0.5*(S.south)+0.5*(dleftmid)$) to[->, bend right, thick, dotted] (S.south);

\draw [thick] (drightmid) to[bend right, thick, dotted] ($0.5*(C.south)+0.5*(drightmid)$);
\draw [to] ($0.5*(C.south)+0.5*(drightmid)$) to[->, bend left, thick, dotted] (C.south);
\draw [thick](drightmid) to[bend left, thick, dotted] ($0.5*(G.south)+0.5*(drightmid)$);
\draw[to] ($0.5*(G.south)+0.5*(drightmid)$) to[->, bend right, thick, dotted] (G.south);

\end{scope}
\end{tikzpicture}
	\end{center}
	\caption{The crooked indifferentiability notion. In the first phase of real world, $\adv_1$ interacts with $f$ and returns an oracle algorithm $\tilde{f}$ (which would be accessed by the construction $C$ in the second phase). In the second phase the random initial value $R$ will be sampled and given it to the construction $C$ and also to $\adv_2$. In ideal world, simulator $S^{\mathcal{F}}$ gets the advise string of the first phase, blackbox access to the subverted implementation $\tilde{H}$ and the initial value $R$. 
	} 
	\label{fig:indiff}
\end{figure} 

\noindent\textsc{Two-Stage Distinguishing Game.} Now we explain the distinguishing game. In the first stage, $\adv_1^h$ outputs $\tilde{H}$ after interacting with a random oracle $h$. Then, a random initial value, $R$, of the hash construction $C$ is sampled. In the real world, $\adv_2$ interacts with the same $h$ of the firsts stage and the construction $C^{\tilde{h}}(R, \cdot)$. In the ideal world, the simulator $S$  gets the advice-string $z$, the initial value $R$ and blackbox access to the subverted implementation $H$ as inputs,\footnote{If an algorithm $S$ gets a blackbox access to the oracle algorithm $H$ as an input, it can compute $H(x)$ by invoking $H$ with input $x$ and responding to the oracle queries of $H$.} and gets oracle access of a random oracle $\mathcal{F}$. Simulator is aimed to simulate $h$ so that behavior of $(h, C^{\tilde{h}})$ is as close as $(S, \mathcal{F})$ to the distinguisher $\adv_2$. \smallskip


\noindent\textsc{Convention on Crooked Distinguishers}: Note that there is no loss to assume that both $\adv_1$ and $\adv_2$ are deterministic (so we skip the notation $r$) when we consider computational unbounded adversary\footnote{$\adv_1$ can fix the best random coin for which the distinguishing advantage of $\adv_2$ is maximum.}. We also assume that $\adv_2$ makes all distinct queries and distinct from the queries made by $\adv_1$. We skip the notation $z$ as an input of $\adv_2$ as it is fixed throughout the game.

\section{Basic Tools of Crooked Indifferentiability}
\label{sec:basic-tools-crooked}
In this section we develop the basic notations and tools used throughout the rest of the paper. The basic technique for designing an efficient simulator in the classical indifferentiability setting, was lazy sampling maintaining consistency. The presence of possibly subverted points in the crooked indifferentiability setting makes direct application of the idea infeasible. We start by defining the following indicator.

\[
d(\alpha, h)= 
\begin{cases}
1,  & \text{if } \tilde{h}(\alpha) \neq h(\alpha) \mbox{ or } \alpha \in \mathcal{D}(z) \\
0,  & \text{otherwise } 
\end{cases}\] 
In other words, $d$ sets value one for an element $\alpha \in \mathcal{D}\setminus\mathcal{D}(z)$, if it is crooked.\footnote{Note that $\alpha \in \mathcal{D}(z)$ was not considered in \cite{C:RTYZ18}. However, it will be later clear that we need this simple modification.} We say $\alpha$ is \emph{problematic} for $h$ if $d(\alpha, h)=1$.

\noindent For every $\alpha \in \mathcal{D}, \beta \in \bool{n}$, we use $h_{\alpha \rightarrow \beta}$ to denote the function which agrees with $h$ on all points possibly except at $\alpha$  which the function $h_{\alpha \rightarrow \beta}$ maps to $\beta$. 
Note that if $h(\alpha) \eq \beta$ then 
$h_{\alpha \rightarrow \beta} \eq h$. Let $D^1(\alpha, h)  \eq \ex_{\beta}(d(\alpha, h_{\alpha \rightarrow \beta}))$ where the expectation is computed under $\beta \sample \bin^n$.

\subsubsection{Bounding the Probability of Problematic Points}
\label{sec:bound-problematic}
 Our first step is to show that even for the subverted implementation of a randomly chosen function, a randomly chosen point of the domain is \emph{not problematic} with high probability.

We follow the notations developed in the previous sections. For $1 \leq j \leq \tau$, let $D^j(\alpha,h)=\mathsf{Ex}_{\beta}(d(\alpha,h_{\tilde{\alpha}_{j}\rightarrow \beta}))$ (average number of crooked point after we resample the output of the $j$th query made by the subverted implementation). 

\begin{lemma}\label{lemma:Dtilde}
	 Let $\alpha \sample \mathcal{D}$, $h \sample \mathsf{H}|_z$. For any $\epsilon$-crooked implementation $H_z$, for all $1\leq j\leq \ell$ 
	\begin{align}\label{dtile:eq}
	\mathsf{Ex}_{\alpha,h}\big({D}^j(\alpha, h)\big) & \leq \epsilon_1 :\eq (\epsilon \pl q_1 2^{-n})
	\end{align}
      \end{lemma}

\subsubsection{Robust points and Robust functions}
\label{sec:robust-points-robust}

\begin{definition}[{\bf Robust points}]
  We say a point $\alpha\in\mathcal{D}$ is robust for function $h$ if
  \begin{enumerate}
  \item $\alpha$ is unsubverted under $h$; $d(\alpha, h) \eq 0$.
   \item If we resample the output of $j^{th}$ query of $\tilde{h}(\alpha)$ then the probability that $h(\alpha)$ becomes subverted is at most $\epsilon_1^{1/2}$. In other words, it holds that
     \begin{align*}
       {D}^j(\alpha, h) \leq \epsilon_1^{1/2}
     \end{align*}
  \end{enumerate}
\end{definition}

\begin{definition}[{\bf Good point-function pair}]
  \label{def:good-pair}
We call a pair $(\alpha', h)$ \emph{good} if $\alpha'$ is queried only by a few robust points. In other words,   
\begin{enumerate}
       \item for all $\alpha \twoheadrightarrow_h \alpha'$ it holds that $\alpha$ is robust for $h$.
  	\item the number of $\alpha$ which queries $\alpha'$ is at most $1/\epsilon_1^{1/4}$.
        \end{enumerate}
        We call the pair  $(\alpha', h)$ bad, if it is not good. Let $G$ be the collection of all such good pairs. For any function $h$, let $G_h :\eq \{\alpha': (\alpha', h) \in G\}$ and $B_h :\eq \{\alpha': (\alpha', h) \not\in G\}$. 
\end{definition}

In what follows we use $\alpha, \alpha' \sample \mathcal{D}$ and $h \sample \mathsf{H}|_z$ whenever these are used as random variables to compute probabilities. Otherwise, these are considered to be fixed elements from their respective domains. 

\noindent The following lemma says that a randomly sampled point-function pair is good with high probability.
\begin{lemma}
  \label{lemma:point-function-good}
$\Pr_{\alpha', h}( (\alpha', h) \mbox{ is bad}) \leq \epsilon_2 :\eq 3\tau \epsilon_1^{1/4} $
\end{lemma}

A function $h$ is said to be robust if $\Pr_{\alpha'} (\alpha' \in B_h) \leq \epsilon_2^{1/2}$ A simple application of Markov inequality proves the following lemma which says that a random function is robust with high probability.

\begin{lemma}
  \label{lemma:robusth}
$\Pr_h(h \mbox{ is not robust}) \leq \epsilon_2^{1/2}$
\end{lemma}

\section{Crooked Indifferentiability of Enveloped XOR construction}
\label{sec:crook-indiff-envel}

In this section we analyze the crooked indifferentiability of security \textsc{EXor}. The construction was analyzed in \cite{C:RTYZ18}, as we shall show, with very restrictive and impractical assumptions. First, we formally define the Enveloped XOR construction with randomized iv. \medskip

\noindent\textsc{Enveloped XOR Construction} or \textsf{EXor}:
We fix two positive integers $n$ and $l$. Let $\mathcal{D} :\eq [l] \times \bin^n$. Let $\mathsf{H}$ be the class of all functions $h : \mathcal{D} \to \bin^n$. 
For every $x \in \bin^n$ and an initial value $R :\eq (r_1, \ldots, r_l) \in (\bin^n)^l$, we define 
$$g_R(m) \eq \bigoplus_{i \eq 1}^{l} h(i, m \oplus r_i)\  \mbox{ and }\ \ \textsf{EXor}(R, m) \eq h(0, g_R(m)).$$ 
Our main result in this section is the following theorem.
\begin{theorem}
  \label{thm:exor}
Let $l\geq n$ and $h : \mathcal{D} \to \bin^n$ be a family of random functions and
$\textsf{EXor}:\bool^{n}\to\bool^{n}$ be the enveloped-xor construction.  Then there exists a simulator $S$ such that for all $(\kappa,\tau, \epsilon)$ crooked distinguisher $\adv=(\adv_1,\adv_2)$
\begin{align*}
  \mbox{\bf Adv}_{\adv,(\textsf{EXor},h)}^{\mbox{\sf crooked-indiff}}\leq  2\epsilon q_2 \pl 2q_2(q_1 \pl q_2)/2^n \pl \sqrt[8]{3\tau\left(\epsilon+\frac{q_1}{2^n}\right)} \pl \frac{1}{2^n}
 \end{align*}
where $q_1$ and $q_2$ are the total number of queries made by the first and the second stage adversaries respectively. The simulator makes one query to the random oracle $\mathcal{F}$ and makes $l$ many calls to the subverted implementation $\tilde{h}$.
\end{theorem}
Restricting $q_1,q_2<<2^n$ $q_2<<1/\epsilon$, and $\tau q_1 << 2^{n/8}$ and $\tau \epsilon << 2^{n/8}$, we make the crooked-indifferentiability advantage of any adversary negligible.

\subsection{Proof of Theorem \ref{thm:exor}}

\subsubsection{A brief detour: classical indifferentiability simulator for {\sf EXor}}
\label{sec:brief-deto-class}
Before describing the crooked indifferentiability simulator, we would like to briefly recall the principle behind the indifferentiability simulator and proof principles behind {\sf EXor} construction in the classical setting. \\
The goal of the simulator is to simulate each $h(i,\cdot)$ honestly so that for every queried message $m$, it holds that $\textsf{EXor}(R, m) \eq \mathcal{F}(m)$ for all queried $m$. Without loss of generality, assume that whenever the adversary makes queries $h(i, x)$ for $i > 0$, it also makes queries $h(j, x \oplus {r}_i \oplus {r}_j)$ for all $j > 0$ simultaneously. In other words, it makes a batch query of the form $({h}(j, m \oplus {r}_j))_{1 \leq j \leq l}$ for some $m \in \bin^n$. We simply say that the adversary $\adv$ queries $m$ to $g_R$ and obtains responses $({h}(j, m \oplus {r}_j))_{1 \leq j \leq l}$.

\noindent On receiving a batch query $g_R(m)$, the simulator will honestly sample outputs for the corresponding $h(i,m\xor R_i)$ queries for all $i\in (l]$, and compute $g_R(m)$ by xoring those sampled outputs. Also, the simulator will save the queried $m$ along with the computed $g_R(m)$ in a list $L$. For a $h(0,x)$ query, the simulator will first search in $L$, whether for some $m$, it has given $x=g_R(m)$ as output. If yes, the simulator simply returns $\mathcal{F}(m)$. If no such entry exists, the simulator samples an output $z$ uniformly at random and returns $z$.

\noindent Now, we briefly recall how the indifferentiability is proved for this simulator. There are two bad events.
\begin{itemize}
\item for distinct $m,m'$, it holds that $g_R(m)=g_R(m')$. In this case, the simulator, on query $h(0,g_R(m))$ can not be consistent with both $\mathcal{F}(m)$ and $\mathcal{F}(m')$ with any significant probability.
  \item For a batch query $g_R(m)$ the output is such that it matches with a previous $h(0,.)$ query. In this case, the simulator has already given output to the $h(0,.)$ query which, with all but negligible probability, is not equal to $g_R(m)$. 
\end{itemize}

One can indeed summarize these bad events as one; $g_R(m)\in E$, where $E$ is the set of $h(0,.)$ queries made by the adversary.

\subsubsection{The Simulator for Crooked Indifferentiability}
\label{sec:simul-crook-indiff}

We now describe the simulator in the crooked indifferentiability setting. The same simulator is used in  \cite{C:RTYZ18}. We define the functions, when instantiated with subverted algorithms.
$$\tilde{g}_R(m) :\eq  \bigoplus_{i \eq 1}^{l} \tilde{h}(i, m \oplus \textsf{r}_i) \mbox{ and } \widetilde{\textsf{EXor}}(R, x) \eq \tilde{h}(0, \tilde{g}_R(x)).$$
In other words, if we express 
${\textsf{EXor}}(R, x)$ as ${\textsf{EXor}}^h(R, x)$ then
$\widetilde{\textsf{EXor}}(R, x)$ represents ${\textsf{EXor}}^{\tilde{h}}(R, x)$. \medskip

Note, here the main goal of the simulator is different. It needs to simulate $h \sample \textsf{H}$ as honestly \footnote{perfectly simulating a random function. If the responses are already in the list it returns that value, otherwise, it samples a fresh random response and includes the input and output pairs in the list.} as possible such that $\widetilde{ \textsf{EXor}}(R, m) \eq \mathcal{F}(m)$ for all queried $m$. Thus the simulator needs to ensure that the output of the random oracle matches with the \emph{subverted implementation} of $\textsf{EXor}$.

The simulator maintains a list $L$ of pairs $(\alpha, \beta)$ to record $h(\alpha) \eq \beta$ for $\alpha \in \mathcal{D}$ and $\beta \in \bin^n$. It also maintains a sub-list $L^A \subseteq L$ consisting of all those pairs which are known to the distinguisher. Both lists are initialized to $z$ (the advice-string in the first stage which we fix to any tuple of $q_1$ pairs). $L_0=L_0^A=z$. Now we describe how the simulator responds.

\begin{enumerate}
\item (Query $h(0, w)$) We call this query a Type-1 Query.  Type-1 Queries are returned honestly. If $((0, w),y)\in L$ for some $y$, the simulator returns the same $y$. Otherwise, it samples $y$ uniformly from $\bin^n$, updates the list $L$ and $L^A$, and returns $y$. \smallskip

\item (Query $g_R(m)$) We call this Type-2 Query. For a query $g_R(m)$ (i.e. batch query) the simulator computes $\tilde{h}(\alpha_j)$ for all $j$, one by one by executing the subverted implementation $\tilde{H}$, where $\alpha_j \eq (j, m \oplus R_j)$. During this execution, simulator responds honestly to all queries made by the subverted implementation and updates the $L$-list by incorporating all query responses of $h$. However, it updates $L^A$ list only with $(\alpha_j, h(\alpha_j))$ for all $j$. Let  $\tilde{\tt g} :\eq \bigoplus_j \tilde{h}(\alpha_j)$. If $(0, \tilde{\tt g})  \in \mathcal{D}(L)$, the simulator {\bf aborts} . If the simulator does not abort, it makes a query $\mathcal{F}(m)$ and adds $((0, \tilde{\tt g}), \mathcal{F}(m))$ into the both lists $L$ and $L^A$. 
\end{enumerate}

\noindent For ${h}(0, w)$ made by $\adv_2$ where $w \eq \tilde{g}_R(m)$ for some previous query $m$ to $g_R$, the simulator responds as $\mathcal{F}(m)$. \smallskip

\noindent\textsc{Cautionary Note}. Even though $\mathcal{F}$ is a random oracle, we cannot say that the probability distribution of the response of $(0, \tilde{\tt g})$ in the ideal world is uniform. Note that, the adversary can choose $m$ after making several consultations with $\mathcal{F}$. In other words, $m$ can be dependent on $\mathcal{F}$. For example, the adversary can choose a message $m$ for which the last bit of $\mathcal{F}(m)$ is zero. Thus, the response for the query $(0, w)$ always has zero as the last bit (which diverts from the uniform distribution). \medskip

\noindent\textbf{Transcript}: Now we describe what is the transcript to the distinguisher and for the simulator in more detail. First, we introduce some more relevant notations.
\begin{enumerate}
\item Let $L^F$ denote the set of all pairs $(m' , \tilde{z})$ of query response of $\mathcal{F}$ by $\adv_2$. \smallskip

\item Let $L^{g}$ denote the set of all pairs $(m, \beta^l)$ of query response of $g_R$ oracle (batch query) made by $\adv_2$ to the simulator where $\beta^l :\eq (\beta_1, \ldots, \beta_l)$ and $\beta_j \eq h(j, m \oplus R_j)$ for all $j$. According to our convention all these $m$ must be queried to $\mathcal{F}$ beforehand. \smallskip

\item As we described, we also have two lists, namely $L$ and its sublist $L^A$, keeping the query responses of $h$ oracle.  
\end{enumerate}

Now we define the transcript and partial transcript of the interaction. 
We recall that $q_1$ is the number of queries in the first stage and $\adv_2$ is a $(q_F, q_2)$-query algorithm. Let $q=q_2+q_F$ For any $1 \leq i \leq q $, we define the partial transcript of $\adv$ and the simulator as $\tau_i^A :\eq (L^F_i, L^A_i)$ and $\tau_i^S :\eq (L_i, L^g_i)$ respectively, where $L^F_i, L^A_i, L_i, L^g_i$ denote the contents of the corresponding lists just before making $i$th query of the distinguisher. 
So when, $i \eq 1$, $L^A_1 \eq L_1 \eq z$ and the rest are empty and when $i \eq q \pl 1$, these are the final lists of transcripts. 
Let $\tau_i :\eq (\tau^A_i, \tau^S_i)$ and $\tau :\eq (\tau^A, \tau^S)$ denote the joint transcript on $i$th query or after completion respectively.
As the adversary is deterministic, the simulator is also deterministic for a given $h$ and $\mathcal{F}$, and we have fixed $z$, a (partial) transcript is completely determined by the choice of $R$, $h$ and $\mathcal{F}$ (in the ideal world). 
We write $(R, h, \mathcal{F}) \vdash \tau_i^S$ if the transcript $\tau_i^S$ is obtained when the initial value is $R$, the random oracles are $\mathcal{F}$ and $h$. We similarly define $(R, h,\mathcal{F}) \vdash \tau_i^A$ and $(R, h,\mathcal{F}) \vdash \tau_i$. \smallskip

\noindent\textbf{Bad Event}. We can now define the bad events. Let $\bad_i$ denote the event that $i$th query is to $g_R(m)$ for some message $m$, and either the simulator aborts (implying $(0,\tilde{g}_R(m))\in L_i$)  or $(0, \tilde{g}_R(m))$ is a crooked point for $h$. Let $\bad \eq \vee_i \bad_i$.   \medskip

\begin{observation}
  Given that $\bad$ does not hold, for all queries $m$ we have $$\widetilde{ \textsf{EXor}}(R, m) \eq \mathcal{F}(m).$$ 
\end{observation}

From now on, we discuss how to bound the probability of the bad events.
\subsection{Techniques of \cite{C:RTYZ18}}
\label{sec:critic}

\noindent\textbf{Overview of the techniques in \cite{C:RTYZ18}}. 
We assume, without any loss of generality that the second stage adversary $\adv_2$ queries $m$ to $\mathcal{F}$ before it queries to $g_R$ oracle. In addition, like before, we assume that it makes batch queries.

For every query number $i$, we define a set $E_i :\eq \mathcal{D}(L_i) \cup \mathcal{C}^h$ where $\mathcal{C}^h$ is the set of all crooked elements for $h$. The event $\bad_i$ holds if and only if $(0, \tilde{g}_R(m_i)) \in E_i$ where $m_i$ denotes the $i$th query of $\adv$ (made to $g_R$ oracle of the simulator). So, the crooked indifferentiable advantage is bounded by $\sum_{i \eq 1}^{q_2} \Pr(\tilde{g}_R(m_i) \in E_i)$. From the definition it is clear that $|\mathcal{D}(L_{i})| \leq q_1 \pl i$. Moreover, $|\mathcal{C}^h|/2^n$ is negligible for every fixed function $h$ as $\adv$ can crook at most negligible fraction of inputs. Motivated from the above, the authors wanted to show that the distribution of $\tilde{g}_R(m_i)$ is almost uniform. They proposed the following theorem. \smallskip

\noindent(\textbf{Theorem 5} from \cite{C:RTYZ18}). Let $h \sample \mathsf{H}|_z$. With overwhelming probability (i.e., one minus a negligible amount) there exists a set $\mathcal{R}_z \subseteq (\bin^n)^l$ and for every $i$, a set of transcripts $\mathcal{T}^A_i$ (before $i$th query) such that for all $R \in \mathcal{R}_z$, $\tau_i :\eq (L_i^F, L_i^A) \in \mathcal{T}^A_i$, 
and $m \not\in \mathcal{D}(L^g_i)$, 
$$\Pr_h((0, \tilde{g}_\textsf{R}(m) ) \in E_i\ |\ (R, \mathcal{F}, h) \vdash \tau_i) \leq \mathrm{poly}(n) \sqrt{|E_i|} \pl \negl(n). \medskip$$

The authors claimed that crooked indifferentiability of ${\sf EXor}$ can be derived from the above theorem. To describe the issues we need to dive into the main steps of proving the above theorem. In the first step, the authors assumed $i \eq 1$ and later with a very sketchy argument justified how it works for general $i$. Note that $\tau_1$ contains only the information of $z$. As $h \sample \mathsf{H}|_z$ is sampled honestly, it is enough to simply focus on bounding the probability of the event 
$(0, \tilde{g}_\textsf{R}(m))  \in \mathcal{D}(z) \cup \mathcal{C}^h$ for any $m$.  The main idea is to show that $\tilde{g}_\textsf{R}(m)$ behaves close to the uniform distribution over $\bin^n$.  Thus the above probability would be negligible as $q_1/2^n$ and $|\mathcal{C}^h|/2^n$ is negligible.

By using Markov inequality, authors are able to identify a set of overwhelming amount of pairs $(R, h)$, called {\em unpredictable} pair, such that for any unpredictable $(R,h)$ all $m$, there exists an index $i$ such that 
\begin{enumerate}
	\item $D^1(\alpha_i, h)$ is negligible and 
	\item \label{step:no-query} $\alpha_j \not\twoheadrightarrow_h\alpha_i$ for all $j \neq i$, where $\alpha_j \eq m \oplus R_i$. 
\end{enumerate}

Thus, if we resample $\beta \eq h(\alpha_i)$ then with overwhelming probability $\tilde{h}_{\alpha_i \to \beta}(\alpha) \eq h_{\alpha_i \to \beta}(\alpha)$ (i.e. not crooked and returned a random value) and all corresponding values for indices $j$ different from $i$ will remain same. So, $\tilde{g}_R(m) \eq \beta \pl A$ where $A$ does not depend on choice of $\beta$. Thus, the modified distribution is close to uniform (as almost all values of $\beta$ will be good). In particular the authors made the following claim: 

\begin{Claim}\label{claim1:old}
Under the modified distribution (i.e. after resampling),  $\Pr(\tilde{g}_R(m) \in E_1) \leq q_1/2^n \pl \epsilon \pl p_n$ where $p_n$ denotes the probability that a random pair $(R, h)$ is not unpredictable.
\end{Claim}
As the choice of $i$ depends on the function $h$ and so a new rejection resampling lemma is used to bound the probability of the event under the original distribution (i.e. before resampling). 
\begin{lemma}[Rejection Resampling \cite{C:RTYZ18}]
	Let $X :\eq (X_1, \ldots, X_k)$ be a random variable uniform on $\Omega=\Omega_1\times \Omega_2\times\cdots\times \Omega_k$. Let $A:\Omega\to (k]$ and define $Z= (Z_1, \ldots, Z_k)$ where $Z_i \eq A_i$ except at $j \eq A(X^k)$ for which $Z_j$ is sampled uniformly and independently of remaining  random variables. Then for any event $S \subseteq \Omega$, it holds that
	\begin{align*}
	|S|/|\Omega| & \leq \sqrt{k \Pr(Z \in S) }
	\end{align*}
\end{lemma}
With this rejection resampling result and the Claim \ref{claim1:old}, the authors concluded the following under original distribution: 
$$\Pr_{h*}(\tilde{g}_R(x) \in E_1) \leq \sqrt{l \cdot \Pr_{\mbox{resampled }h}(\tilde{g}_R(x) \in E_1)} \leq \sqrt{l \cdot (q_1/2^n \pl \epsilon \pl p_n)}.$$ 
\subsection{Issues with the technique of \cite{C:RTYZ18}}\label{sec:issues-with-approach}
Now we are ready to describe the issues and the limitations of the techniques in \cite{C:RTYZ18}. To prove the general case (i.e. for any query), authors provides a very sketchy argument. It seems that authors argued that with an overwhelming probability of realizable transcript $\mathcal{T}$ and for all $\tau \in \mathcal{T}$, $\Pr(\tilde{g}_R(m_i) \in E_i \ |\ \tau)$ is negligible. \smallskip

\noindent\textbf{Inconsistency for Multiple Queries: Controlling query dependencies for the same index.}

Recalling the notion of unpredictable $(R,h)$ we see that the resampling is done on an index $i$, that is honest ($\tilde{h}(i,m\xor R_i)=h(i, m\xor R_i)$, and $h(i,m\xor R_i)$ is not queried by $h(i,m\xor r_j)$ for any other $j$. From here, the authors argued that the transcript of the interaction remains same, if we resample at such $i$.
However, if the adversary can check consistency for multiple messages, this is not sufficient. 

Consider a $(q,\tau)$ subverted implementation $\tilde{h}$ that is $\epsilon$-crooked. We construct a  $(q,\tau+1)$ subverted implementation $\hat{h}$ as follows. For every possible input $(i,x)$, $\hat{h}$ simulates $\tilde{h}$. After the $\tau$ many queries made by $\tilde{h}$, $\hat{h}$ makes an additional query on $h(i,x\xor 1^n)$, if not already made. Finally $\hat{h}$ outputs what $\tilde{h}$ outputs. Clearly,  $\hat{h}$ is $\epsilon$-crooked. \\
The distinguisher  makes two batch queries queries,  $\tilde{g}_R(m\xor 1^n)$, and  $\tilde{g}_R(m)$. The simulator, while simulating $\tilde{g}_R(m\xor 1^n)$ responds to all the queries made by $\tilde{h}(i,m\xor 1^n\xor R_i)$, and in particular the value of $h(i,m\xor R_i)$ is now gets fixed. Now consider doing the resampling for the responses of the batch query $\tilde{g}_R(m)$. All the $h(i,m\xor R_i)$ has been fixed. Thus one can not find an index $i$ such that $h(i,m\xor R_i)$ has never been queried. Thus if we resample at any index $i$, the transcript gets changed. Hence \emph{the claim that for all unpredictable $(R,h)$, for all $m$, an index $i$ exists on which the resampling can be done without affecting the transcript is false.}  \smallskip 

\noindent\textbf{The bad event $E_i$ depends on the function $h$.}Claim \ref{claim1:old} says that $\Pr_{\mbox{resampled }h}(\tilde{g}_R(x) \in E_1)$ is small because $\tilde{g}_R(x)$ is uniformly distributed under resampling distribution of $h$ and size of $E_1$ is negligibly small. But, that disregards the fact that the crooked set of $h$ may depend on the function family $h$ and hence $E_1$ is not independent of $\tilde{g}_R(x)$. In particular, one cannot upper bound the $\Pr(\tilde{g}_R(x) \in E_1)$ as $|E_1|/2^n$ . In other words, the Claim \ref{claim1:old} need not be true. This is one of the crucial observation which actually makes the crooked security analysis a bit complex. \smallskip

\noindent\textbf{The number of queries to $\mathcal{F}$ is essential.} Another incompleteness of the proof of \cite{C:RTYZ18} comes from the fact that the analysis does not consider the $\mathcal{F}$ queries of the distinguisher. The bound is almost vanishing if $q_1 \eq 0$ and $q_2 \eq 2$ and has no crooked point. However, a distinguisher can search for $m \neq m'$ such that $\mathcal{F}(m) \eq \mathcal{F}(m')$. Conditioned on collision at the final output, the event $g_R(m) \eq g_R(m')$ holds with probability about $1/2$. However, for the honest simulation of all $h$ values, $g$ value will collide with very low probability. Hence, if the adversary can make $2^{n/2}$ many queries to $\mathcal{F}$, the above consistency can be forced. Hence the probability upper bound of Theorem 5 of \cite{C:RTYZ18} can not be independent of the number of queries made to $\mathcal{F}$.

\noindent\textbf{The gap in the technique.} The reason, these issues have not cropped up in the proof is the fact, that the authors \emph{did not prove} the simulator to be consistent. In particular, they did not show that the bad events they considered, are complete. To prove indifferentiability, it is required that conditioned on not Bad, the real and the ideal games are indistinguishable. While there are several techniques (like H-coefficient technique, or game-playing technique), there is no formal argument on why Theorem $5$ proves crooked indifferentiability.  \smallskip

\subsection{Our Proof of Theorem~\ref{thm:exor}}
\label{sec:tools-techniques}
Our objective is to show that for every message, there is an index where we can resample without affecting the (partial) transcript. \\
\noindent\textbf{Critical Set.} We construct a set $\mathcal{G}^*$ which we call a critical set. A pair $(R,h)\in\mathcal{G}^*$ if
\begin{enumerate*}
\item $h$ is robust,
\item for every $m$ there exists $i$ so that
  \begin{enumerate*}
  \item  $(\alpha_i :\eq (i, m \oplus R_i), h)$ is good.
   \item $\forall~j<i$, $\alpha_j \not\twoheadrightarrow \alpha_i$.
   \end{enumerate*}
 \end{enumerate*}
 We call the smallest index $i$ that satisfies the above two condition, ``index of interest'' for $(m,R,h)$.
 
\noindent If we can identify a critical set $G^*$, the following nice condition holds. \emph{For every fixed $(R, h) \in \mathcal{G}^*$ and for every $m$ there exists the index of interest, $i$  such that $(\alpha_i, h)$ is good (see Definition \ref{def:good-pair}) where $\alpha_i \eq (i, m \oplus R_i)$.} 

Thus, for every message, we get an $\alpha_i=m\xor R_i$ such that for all $\alpha$ such that $\tilde{h}(\alpha)$ queries $h(\alpha_i)$, is unsubverted under $h$ and remains unsubverted if we resample $h(\alpha_i)$.

\noindent The following lemma says that for a uniform random string $R$ and a randomly chosen function $h$, with high probability $(R,h)$ is in the critical set. Recall from Lemma \ref{lemma:Dtilde}, Lemma \ref{lemma:point-function-good}, that we use the notation $\epsilon_1=\epsilon+q_12^{-n}$ and $\epsilon_2=3\tau\epsilon_1^{1/4}$.

\begin{lemma}
  \label{lemma:goodRh}
  Let $\epsilon_2 \leq 1/16$ and $\ell \ge n$. It holds that  $\Pr_{R,h}((R, h) \not\in \mathcal{G}^*) \leq p_1:=\epsilon_2^{1/2} \pl 2^{-n}$
 
\end{lemma}

\noindent{\bf Resampling on index of interest does not affect the transcript.}
Fix $(R,h)$ from the critical set. Fix a message $m$ and let $i$ be the index of interest. Our objective is to show that a transcript remains unchanged when $h(\alpha_i)$ is resampled. For that, our next step is to show the following.  We can identify a set $\mathcal{S}$  (of size close to $2^n$) such that for all $\beta \in \mathcal{S}$ and $h_\beta \eqdef h_{\alpha_i \to \beta}$, it holds that $\tilde{h}(x) \eq \tilde{h_\beta}(x)$ for all $x \neq \alpha_i$.\\ We define $\mathcal{S} \eq \{ \beta: d(\alpha, h_{\tilde{\alpha}_j \to \beta}) \eq 0 \forall~ j, \forall ~\alpha\in \mathcal{Q}^h_{\twoheadrightarrow \alpha_i}\}$ where $\mathcal{Q}^h_{\twoheadrightarrow \alpha_i} :\eq \{\alpha: \alpha \twoheadrightarrow \alpha_i\}$. \smallskip 

\noindent\textbf{Lower Bounding the size of $S$.}
From Definition \ref{def:good-pair}, as $(\alpha_i,h)$ is good, the size of the set $\mathcal{Q}^h_{\rightarrow \alpha_i}$ is at most $1/\epsilon_1^{1/4}$. Moreover, for every such $\alpha\twoheadrightarrow_h\alpha_i$, we have $d(\alpha, h) \eq 0$ and $\tilde{D}(\alpha, h) \leq \epsilon_1^{1/2}$.  So, there are at most $2^n \times \epsilon_1^{1/2}$ many $\beta$ for which there exists some $j\in(l]$ with $d(\alpha, h_{\tilde{\alpha}_j \to \beta}) \eq 1$. Using union bound for all $\alpha$ that queries $\alpha_i$, the size of the set $\bin^n \setminus \mathcal{S}$ is at most $\frac{2^n \epsilon_1^{1/2} }{\epsilon_1^{1/4}} \eq 2^n \cdot \epsilon_1^{1/4}$. \smallskip
\begin{observation}\label{obs:h-hbeta}
$\tilde{h}(x) \eq \tilde{h_\beta}(x)\ \forall x \neq \alpha_i$. 
\end{observation}
Clearly, $\tilde{h_\beta}(x)$ can be different from $\tilde{h}(x)$, only if $x \twoheadrightarrow \alpha_i$. However, for all such $\alpha_i$ and for all $\beta \in \mathcal{S}$, we have shown that both $d(\alpha, h) \eq d(\alpha, h_\beta) \eq 0$. Hence, $\tilde{h_\beta}(\alpha) \eq h_\beta(\alpha) \eq h(\alpha) \eq \tilde{h}(x)$. So from Observation~\ref{obs:h-hbeta},
\begin{enumerate}
\item for all $m ' \neq m$, $\tilde{g}_R^h(m') \eq \tilde{g}_R^{h_\beta}(m')$.
  \item  $d((0, g), h) \eq d((0, g), h_\beta)$ (the crooked set of $h_\beta$ for zero index is same for $h$ for all such $\beta$).
\end{enumerate}

\noindent\textbf{Uniformity conditioned on transcript: Index of Interest is independent of $\beta$.}  The last step is to ensure that the point of resampling is independent of $\beta$. In other words, we need to show that the index of interest is also independent of $\beta$. This will imply, for every message $m$, we can identify an index $i$ such that $\tilde{g}_R^{h_\beta}(m) \eq \beta \oplus \tilde{g}^h_R(m) \oplus h(\alpha_i)$ holds for the fixed transcript.

\begin{lemma}
  \label{lemma:index-independent-beta}
Fix a good $(R,h)$. Fix a message $m$. Let $i$ be the index of interest for $(m,R,h)$.  Let $\alpha_i=(i,m\xor R_i)$. Then the following holds,
\begin{enumerate}
\item  $\alpha_i\notin B_{h_\beta}$.
\item for all $j<i$, $\alpha_j\not\twoheadrightarrow \alpha_i$ where $\alpha_j=(j,m\xor R_j)$.
\item $i$ is minimum index that satisfies the above two condition for $(R,h_\beta)$ and $m$.  
\end{enumerate}
 \end{lemma}
 
\noindent Finally, we are ready to state the main proposition

\begin{proposition}\label{prop:key}
For any partial transcript for adversary $\tau _j:\eq (L^F, L^A)$, let $(R, h, F) \vdash \tau_j$ such that $(R, h) \in \mathcal{G}^*$. For every $m \not\in \mathcal{D}(L^g)$, there is a set $\mathcal{S}$ of size at least $2^n(1 - \epsilon_1^{1/2})$ such that for all $\beta \in \mathcal{S}$, $(R, h_\beta, \mathcal{F}) \vdash \tau_j$ and $\tilde{g}_R^{h_\beta}(m) \eq \beta \oplus \tilde{g}^h_R(m) \oplus h(\alpha_i)$, where $i$ is the index of interest of $(m,R,h)$. Moreover, the crooked sets of $h$ and $h_\beta$ are same. $C_0^{h_\beta } \eq C_0^h$ (the crooked sets are same). 

Assuming $\epsilon_1^{1/2} \leq 1/2$, and for all $(R, h, \mathcal{F}) \vdash \tau$ with $(R, h) \in \mathcal{G}^*$, we have 
$$\Pr_\beta((0, \tilde{g}_R^{h_\beta}(m) )\in \mathcal{D}(L_0) \cup C^{h_\beta} \wedge (R, h_\beta, F) \vdash \tau) \leq  2\epsilon  \pl 2(q_1 \pl i)/2^n$$ 
\end{proposition}

\subsubsection{ Bad Events and the Probability Bound.}
We describe the bad events in terms of transcript notations as we did before. Following the transcript notation in Section \ref{sec:critic}, recall that the $\mbox{\sc Bad}$ is raised by two events, that $i$th query is $m$ to $g_R$ oracle for which  $(0, \tilde{g}_R(m)) \in \mathcal{D}(L_i)$ (denotes as $\mbox{\sc Bad}_1$) or $(0, \tilde{g}_R(m)) \in C^h$ (a crooked point for $h$ and rest is same as $\mbox{\sc Bad}_2$). Let $\mbox{\sc Bad} \eq \vee_i \mbox{\sc Bad}_i$. Thus, the distinguishing advantage is bounded by $\Pr(\vee_i \mbox{\sc Bad}_i)$. We need to compute the probability given the randomness of $R, h$ $F$ such that $(R, h, F) \vdash \tau_i^A$ (transcript of the adversary). We first bound $\Pr(\mbox{\sc Bad} \wedge (R, h) \in \mathcal{G}^*)$. By using Proposition \ref{prop:key}, for every $i$, $\Pr(\mbox{\sc Bad}_i | (R, h, F) \vdash \tau_i, (R, h) \in \mathcal{G}^*) \leq 2\epsilon \pl 2(q_1 \pl i)/2^n$. Hence by summing over all $i$ and using the bound of probability of not realizing $\mathcal{G}^*$, we get
\begin{align*}
  \Pr(\mbox{\sc Bad}) &\leq 2\epsilon q_2 \pl 2q_2(q_1 \pl q_2)/2^n \pl p_1\\
            &= 2\epsilon q_2 \pl 2q_2(q_1 \pl q_2)/2^n \pl \sqrt[8]{3\tau\left(\epsilon+\frac{q_1}{2^n}\right)} \pl \frac{1}{2^n}
\end{align*}

Note that $q_2$ denotes the total number of queries of $\adv_2$ made to both the simulator and $\mathcal{F}$ (in our convention adversary makes all $\mathcal{F}$ queries to $g_R$ of the simulator). This finishes the proof of Theorem \ref{thm:exor}.

\section{Games for the proof of Theorem \ref{thm:exor}}
\label{sec:crook-indiff-mult}

In this section, we show the game transitions in the proof of Theorem \ref{thm:exor}, concluding that the crooked-indifferentiability of {\sf EXor} construction is indeed bounded by the probability of the bad events we considered. We assume that distinguisher makes all $\mathcal{F}(m)$ queries to $g_R(m)$ in case it is not queried. However, it would be done after all queries are done. Note, there is no loss to release all these $m$ values after the original distinguisher finishes the queries. The query complexity of the distinguisher increases by at most $\ell$ times.

\subsection{Game Transitions}
\label{sec:game-transitions}

Our crooked-indifferentiability proof relies on three intermediate game, denoted by $G_0, G_1,$ and $G_2$. We start with the real game $G_0 :\eq (h, C^{\tilde{h}})$. There are two public interfaces for the adversary to query. The first one is $\mathcal{O}_h$, which can be used to interact with the function $H_l$. The other one is $\mathcal{O}_C$, which can be used to compute $C^{\tilde{h}}$. For ease of explanation, we add two internal subroutines, one for computing $\tilde{g}_R$ and the other for evaluating $\tilde{h}(0,\cdot)$.

\begin{figure}[!htb]
	\centering
	\fbox{\scalebox{0.85}{
			\begin{pchstack}
				\begin{pcvstack}
					Game $(h,C^{\tilde{h}})$ \\\\
					\procedure{$\mathcal{O}_h(i,x)$ ($i\in [\ell]$)}{%
						\pcln \pcreturn h(i,x) \\
					}
					\pcvspace
					\procedure{$\tilde{h}_0(x)$}{%
						\pcln \pcfor \mbox{all queries } (i, \alpha) \mbox{ made by } \tilde{h}\\
						\pcln \mbox{Feed } h(i,\alpha)\\
						\pcln  z=\tilde{h}(0, x)\\
						\pcln \pcreturn z
					}
				\end{pcvstack}
				\pchspace
				\begin{pcvstack}
					\procedure{$\mathcal{O}_C(m)$}{
						\pcln S_m= \tilde{g}_R(m)\\  
						\pcln  z=\tilde{h}(0, S_m)\\
						\pcln \pcreturn z\\
					}
					\pcvspace
					\procedure{$\tilde{g}_R(m)$}{%
						\pcln Sum=0^n\\
						\pcln \pcfor j=1 \mbox{ to } \ell \pcdo\\
						\pcln \pcind \mbox{Run }\tilde{h}(j, m\xor R_j)\\
						\pcln \pcind\pcfor \mbox{ all queries} (i, \alpha) \mbox{ made by } \tilde{h}\\
						\pcln \pcind \pcind \mbox{Feed } h(i,\alpha)\\
						\pcln \pcind u_j=\tilde{h}(j, m\xor R_j)\\
						\pcln \pcind Sum=Sum\xor u_j\\
						\pcln \pcendfor\\
						\pcln \pcreturn Sum
					}
				\end{pcvstack}
			\end{pchstack}
	}}
	\caption{Game \pcnotionstyle{Real}}
	\label{fig:gamereal}
\end{figure}

\noindent \textbf{Game ${\bf G}_0$. } In this game, we modify the  $\mathcal{O}_h$ subroutine. For every $(j,x)$ query we recover the message $m=x\xor r_j$, and precompute the response of all the  $(i,m\xor r_i)$ queries. Further, we compute the value of $\tilde{z}=\tilde{h}(\tilde{g}_R(m))$. These precomputations do not change the output for any of the query. Hence

\begin{align*}
    \Pr[\adv^{h,C^{\tilde{h}}}=1]=\Pr[\adv^{{\bf G}_0}=1]
\end{align*}

\begin{figure}[!htb]
	\centering
	\fbox{\scalebox{0.85}{
			\begin{pchstack}
				\begin{pcvstack}
					Game $G_0$ \\\\
					\procedure{$\mathcal{O}_h(j,x)$ ($j\in [\ell]$)}{%
						\pcln \pcif (j,x,y)\in L\pcind\pcreturn y\\
						\pcln \pcif j>0 \\
						\pcln \pcind m=x\xor R_j\\
						\pcln S_m= \tilde{g}_R(m)\\
						\pcln  \tilde{z}=\tilde{h}(0, S_m)\\
						\pcln \pcind (j,x,y)\leftarrow L\\
						\pcln \pcind \pcreturn y\\
						\pcln \pcif j=0\\
						\pcln \pcreturn h(0,x) \\
					}
					\pcvspace
					\procedure{$\tilde{h}_0(x)$}{%
						\pcln \pcfor \mbox{all queries } (i, \alpha) \mbox{ made by } \tilde{h}\\
						\pcln \mbox{Feed } h(i,\alpha)\\
						\pcln  \tilde{z}=\tilde{h}(0, x)\\
						\pcln \pcreturn \tilde{z}
					}
				\end{pcvstack}
				\pchspace
				\begin{pcvstack}
					\procedure{$\mathcal{O}_C(m)$}{
						\pcln S_m= \tilde{g}_R(m)\\  
						\pcln  z=\tilde{h}(0, S_m)\\
						\pcln \pcreturn z\\
					}
					\pcvspace
					\procedure{$\tilde{g}_R(m)$}{%
						\pcln Sum=0^n\\
						\pcln \pcfor j=1 \mbox{ to } \ell \pcdo\\
						\pcln \pcind \mbox{Run }\tilde{h}(j, m\xor R_j)\\
						\pcln \pcind\pcfor \mbox{ all queries} (i, \alpha) \mbox{ made by } \tilde{h}\\
						\pcln \pcind \pcind \mbox{Feed } h(i,\alpha)\\
						\pcln \pcind u_j=\tilde{h}(j, m\xor R_j)\\
						\pcln \pcind Sum=Sum\xor u_j\\
						\pcln \pcendfor\\
						\pcln \pcreturn Sum\\
					}
				\end{pcvstack}
			\end{pchstack}
	}}
	\caption{Game $G_0$}
	\label{fig:game0}
\end{figure}

\noindent \textbf{Game ${\bf G}_1$.} In this game, we introduce two lists $L_f$ and $L_c$. The entries in both the lists are of the form  $(m,x,z,\tilde{z})$. We also introduce two {\sc Bad} events in the code of $\mathcal{O}_h$ as well as in the code of $\mathcal{O}_C$. Notice that, both the subroutine computes $\tilde{h}(0,\tilde{g}_R(m))$.\\ 
The first bad event ($\mbox{\sc Bad}_1$) happens if $h(0,\tilde{g}_R(m))$ has been set already. This can happen in two ways. The first one is during a previous $h(0,\tilde{g}_R(m'))$
computation for a different $m'$. In that case there is a collision in the output of $\tilde{g}_R$. The second way is via a $\mathcal{O}_h(0,x)$ query (by the distinguisher or the subverted implementations). When queried, such an $x$ was not related to a message.\\
The second bad event happens ($\mbox{\sc Bad}_2$) if $\tilde{h}(0,\tilde{g}_R(m))\neq h(0,\tilde{g}_R(m))$. In other words, $\tilde{g}_R(m)$ is a subverted point for $\tilde{h}(0,.)$.
The final change is in the introduction of the random oracle $\mathcal{F}$. Our intention in this game is to program the $h(0,x)$ as $\mathcal{F}(m)$ if $x=\tilde{g}_R(m)$. Hence after the compuation of $\tilde{g}_R(m)$,  if we find $h(0,\tilde{g}_R(m))$ is not already set ($\mbox{\sc Bad}_1$ did not happen), we set $h(0,\tilde{g}_R(m))=\mathcal{F}(m)$. As we are not changing the previously set values, the transcript is consistent. Thus, the distinguisher's view remains unchanged. Thus we get,

\begin{align*}
   \Pr[\adv^{{\bf G}_0}=1]=\Pr[\adv^{{\bf G}_1}=1]
\end{align*}

Before moving to the next game, we state the significance of $\mbox{\sc Bad}_2$. Looking ahead, in such a situation, the simulator will not be able to ``program'' the output of  $\tilde{h}(0,\tilde{g}_R(m))$ as $\mathcal{F}(m)$, and thus loosing the consistency with the random oracle. \smallskip
\begin{figure}[!htb]
  \centering
  \fbox{\scalebox{0.85}{
      \begin{pchstack}
        \begin{pcvstack}
          Game $G_1$ \\\\
          \procedure{$\mathcal{O}_h(0,x)$ }{%
            \pcln \pcif (*,x,z,*)\in L_f\\
            \pcln \pcind \pcreturn z\\
            \pcln z=h(0,x)\\
            \pcln \mbox{Add the entry}(-,x,z,-)\rightarrow L_f\\
         } 
          %
           \pcvspace
           \procedure{$\tilde{g}_R(m)$}{%
              \pcln Sum=0^n\\
                 \pcln \pcfor j=1 \mbox{ to } \ell \pcdo\\
                 \pcln \pcind \mbox{Run }\tilde{h}(j, m\xor R_j)\\
                  \pcln \pcfor \mbox{all queries } (0, \alpha) \mbox{ made by } \tilde{h}\\
             \pcln z=\mathcal{O}_h(0,\alpha)\\
           \pcln \pcind\pcfor \mbox{ all queries} (i>0, \alpha) \mbox{ made by } \tilde{h}\\
           \pcln \pcind \pcind \mbox{Feed } h(i,\alpha)\\
           \pcln \pcind u_j=\tilde{h}(j, m\xor R_j)\\
           \pcln \pcind Sum=Sum\xor u_j\\
           \pcln \pcendfor\\
           \pcln \pcreturn Sum\\
         }
           \pcvspace
           \procedure{$\tilde{h}_0(x)$}{%
             \pcln \pcfor \mbox{all queries } (0, \alpha) \mbox{ made by } \tilde{h}\\
             \pcln z=\mathcal{O}_h(0,\alpha)\\
              \pcln \pcfor \mbox{all queries } (i, \alpha) \mbox{ made by } \tilde{h}\\
              \pcln \mbox{Feed } h(i,\alpha)\\
             \pcln  \tilde{z}=\tilde{h}(0, x)\\
              \pcln \pcreturn \tilde{z}
            }
           \end{pcvstack}
           \pchspace
           \begin{pcvstack}
              \procedure{$\mathcal{O}_h(j,x)$ ($j>0$)}{%
            \pcln \pcif (j,x,y)\in L\pcind\pcreturn y\\
            \pcln \pcind m=x\xor R_j\\
            \pcln \pcfor i=1 \mbox{ to }\ell\\
            \pcln \pcind \mbox{Add }(i,m\xor R_i,f(i,m\xor R_i)) \mbox{ to } L\\
            \pcln \pcendfor\\
           \pcln  S_m= \tilde{g}_R(m)\\
           \pcln  \pcif (*,S_m,z,*)\in L_f \mbox{ for any }z\\
           \pcln \pcind\mbox{\sc Bad}1=1\\
           \pcln \pcind\pcbox{\tilde{z}=\tilde{h}(0, S_m)} \gamechange{$\tilde{z}=\mathcal{F}(m)$}\\
           \pcln \pcind \mbox{Add the entry}(m,S_m,z,\tilde{z})\rightarrow L_f\\
           \pcln \pcelse\\
           \pcln \pcind\tilde{z}=z=\mathcal{F}(m)\\
            \pcln \pcind \mbox{Add the entry}(m,S_m,z,\tilde{z})\rightarrow L_f\\
           \pcln \pcind \tilde{z'}=\tilde{h}(0, S_m)\\
           \pcln \pcind \pcif \tilde{z}\neq \tilde{z'}\\
           \pcln \pcind[2] \mbox{\sc Bad}2=1\\
           \pcln \pcind[2] \pcbox{\tilde{z}=\tilde{z'}} \gamechange{Do nothing}\\
           \pcln \pcind\pcendif\\
           \pcln \pcind \mbox{Overwrite the entry}(m,S_m,z,\tilde{z})\rightarrow L_f\\
           \pcln \pcendif\\
            \pcln (j,x,y)\leftarrow L\\
            \pcln  \pcreturn y\\
          }
           \end{pcvstack}
          \end{pchstack}
    }}
  \caption{Game $G_1,G_2$ of {\sf EXor} proof. The boxed entries are executed in $G_1$ whereas the highlighted entries are executed in $G_2$.}
  \label{fig:game1}
\end{figure}

\noindent \textbf{Game ${\bf G}_2$.} In this game, we introduce the changes in the computation. First, we (re)program $h(0,\tilde{g}_R(m))=\mathcal{F}(m)$ even if it was previously set. Moreover, we set $\tilde{z}$, the output of $O_c(m)$ query to always be same as $h(0,\tilde{g}_R(m))$. These modifications create changes in the output in two places. The first one is in the case of $\mbox{\sc Bad}_1$. The second is in the case of $\mbox{\sc Bad}_2$, $\tilde{g}_R(m)$ is a subverted point for $\tilde{h}(0,.)$. The rest of the game remains unchanged. As the two games are identical until one of the bad event happens, using the fundamental lemma of game playing proofs, 
\begin{align*}
\left|\Pr[\adv^{{\bf G}_1}=1]=\Pr[\adv^{{\bf G}_2}=1]\right| \leq \Pr[\mbox{\sc Bad}_1\cup\mbox{\sc Bad}_2].  
\end{align*}
\noindent \textbf{Game $(S^{\mathcal{F}},\mathcal{F})$.} It is also easy to see that the output distribution of the game ${\bf G}_2$ is identical to the game $(S^{\mathcal{F}},\mathcal{F})$. 
\begin{align*}
  \Pr[\adv^{{\bf G}_2}=1]=\Pr[\adv^{(S^{\mathcal{F}},\mathcal{F})}=1]
\end{align*}

Finally, collecting all the probabilities, we get,

\begin{align*}
  \Delta_{\adv_2(r, z, R)}\big((h, C^{ \tilde{h} }(R, \cdot))\ ;\ (S^{\mathcal{F}}(H_{z, r}, z, R), \mathcal{F}) \big)\leq \Pr[\mbox{\sc Bad}_1 \cup \mbox{\sc Bad}_2].  
\end{align*}

\section{Crooked-Indifferentiability of Sponge Construction}
\label{sec:new-tool-proving}


\begin{figure}[!htb]
  \centering
   \fbox{\scalebox{0.7}{
      \begin{pchstack}
        \begin{pcvstack}
          \procedure{ Procedure Sponge (Random string $R$, Message $m_1,\cdots,m_\ell$)}{%
            \pcln \label{step:init}x= (x_a,x_c)= R\\
            \pcln \pcfor i=0 \mbox{ to } \left\lceil \frac{\ell}{r}\right\rceil-1 \pcdo\\
            \pcln \pcind (x_a,x_c)=h(x_a\xor m_i,x_c)\\
            \pcln \pcendfor\\
            \pcln \pcfor i=0 \mbox{ to } \left\lceil \frac{h}{r}\right\rceil -1 \pcdo\\
            \pcln \pcind \mbox{Append } x_a \mbox{ to output}\\
            \pcln \pcind (x_a,x_c)=h(x_a,x_c)\\
            \pcln \pcendfor\\
          }
          \end{pcvstack}
        \end{pchstack}}}
  \caption{Sponge Based Construction}
  \label{fig:sponge}
\end{figure}
\subsection{Sponge Construction}
\label{sec:sponge-construction}
We recall the sponge-construction \cite{BertoniDPA07}.
Fix positive integers $r,c$, and let $n=r+c$. Let $h:\bool^n\to\bool^n$ be a function.  The sponge construction $C^h$ maps binary strings of length $\ell$ bits to $s$ bit binary digest.

\subsection{Crooked Indifferentiable Sponge Construction}
\label{sec:crook-indiff-sponge}
To handle subversion, we randomize the sponge construction by setting the IV to be equal to the random string $R$. The rest of the construction is unchanged. Our main result in this section is Theorem~\ref{thm:sponge}.
\begin{theorem}
  \label{thm:sponge}
Let $h:\booln\to\booln$ be a random function and
$C^h:\bool^{\ell}\to\bool^{s}$ be the sponge construction. Let $r$ be the rate part and $c$ be the capacity part of the chain. Then there
exists a simulator $S$ such that for all $(\kappa,\tau, \epsilon)$ crooked distinguisher $\adv=(\adv_1,\adv_2)$
\begin{align*}
  \mbox{\bf Adv}_{\adv,(C,f)}^{\mbox{\sf crooked-indiff}}\leq
  \frac{q^2\tau^2+q_2(\ell+s)\kappa}{2^{c}}+2^r\epsilon q_2(\ell+s)
\end{align*}
where  $q_2$ is the total number of construction queries made by $\adv_2$ and $q$ is the total number of blocks in the queries made by $\adv_2$. 
\end{theorem}

\noindent\textbf{Proof Sketch.} We start from the following observation. If for all construction query made by $\adv_2$, none of the intermediate queries (made $C$to $\tilde{h}$) are subverted, then by the indifferentiability result of \cite{EC:BDPV08}, we achieve crooked-indifferentiability. Hence, we say the {\sc Bad} event occurs, if for some construction query $M$, made by $\adv_2$, the intermediate query is subverted. We bound the {\sc Bad} probability in the real world. By the definition of crooked-indifferentiability, the probability that $\tilde{h}(R)$ is subverted, for a randomly chosen $R$, $h$ is $\epsilon$. By union bound, the probability that for some $m_0\in\bool^r$, $\tilde{h}\left(R\xor (m\|0^c)\right)$ is not equal to $h\left(R\xor (m\|0^c)\right)$ is at most $\epsilon2^r$. Conditioned on the input being ``non-subverted'', the output of $\tilde{h}$ is independently and uniformly distributed. Hence, we get a uniform random chaining value. Repeating the argument, and taking union bound, we get the following. For any message $M$, the probability that one of the $\ell+s$ many queries made by $C^h$ is subverted, is at most $(\ell+s)\epsilon2^r$. Taking union bound over all the queries made by the distinguisher, the probability becomes bounded by $q(\ell+s)\epsilon2^r$. In addition, the probability that the $c$-part of the chaining value is equal to the $c$-part of some queries in the first stage is $\frac{q_2(s+\ell)\kappa}{2^c}$. Given that none of the above two events happens, the distinguishing advantage of any adversary is the same as the standard indifferentiability advantage. For every query $q$, the simulator needs to run the subverted implementation which makes $\tau$ many queries for each invocation. Hence, our simulator makes at most $q\tau$ many queries to the simulator of Bertoni \etal~\cite{eurocrypt/BertoniDPA08}.  Thus the advantage of the adversary is at most $\frac{q^2\tau^2+q_2(\ell+s)\kappa}{2^c}+ 2^r\epsilon q_2(\ell+s)$. For detail proof using game-playing technique, we refer the reader to Section \ref{sec:spongegames}.

\section{Proof of Theorem~\ref{thm:sponge}}
\label{sec:spongegames}

\subsubsection{The Simulator}
\label{sec:simulator}
Our simulator emulates the simulator of Bertoni \etal \cite{EC:BDPV08}. For completeness, we recall the simulator below. Specifically, we recall the following objects used in the proof. 

\subsubsection{The Simulator Graph}
\label{sec:simulator-graph}

The simulator maintains a graph $G$ for recording the interactions for
$f$. The vertex set of the graph is
$V(G)\subseteq \bool^{r}\times\bool^{c}$. We represent a $v\in V(G)$ by an
ordered pair $(v_r,v_c)$ where $v_r\in\bin^r$ and
$v_c\in\bool^{c}$. The (directed) edge set of the graph is represented by
$E(G)$. The simulator also keeps a list $L \subseteq \bool^{c}$. $L$ is used to ensure that the $c$-part of all the responses of the
simulators are unique.

\subsubsection{NewNode Algorithm.} The algorithm Newnode samples a node randomly in the simulator graph. 

\subsubsection{Findpath Algorithm} The Findpath algorithms finds a message $m \in \bool^{\leq \ell}$ such that evaluating $C$ with the random string  $R$ as IV, message $m$ and simulator's
responses so far will generate $x$ as a query to $f$. In other words, $x$ will be a chaining value in the computation of $C_R^h(m)$.

\subsubsection{Simulating $f$ in Stage I}
\label{sec:simulating-I-stage}
We denote our simulator using $\hat{S}$. In the first stage, $\hat{S}$ simulates $f$ honestly. $\hat{S}$ starts with a \emph{local copy} of the simulator graph with all the nodes, but no edges ($E(G)=\emptyset$) and an empty list $L$. When queried
with a new input $x$, the simulator generates $y_c\sample \bool^{c}\setminus L$ and
$y_r\sample \bool^{r}$, creates node $(y_r,y_c)$, and adds an edge
$(x,(y_r,y_c))$ to $E(G)$. $\hat{S}$ updates $L$ by including $y_c$ to the list.

\begin{figure}[!htb]
  \centering
  \fbox{\scalebox{0.85}{
      \begin{pchstack}
        \begin{pcvstack}
               \procedure{ Procedure $\hat{S}_1(x)$ \pccomment{Stage I}}{%
           \pcln \pcif \exists (x,y)\in E(G') ~\pcreturn y\\
           \pcln \pcelse\\
           \pcln~ y_c\sample\bool^{c}\setminus L\\
           \pcln~ y_r\sample\bool^{r} \\
           \pcln~ y=(y_r,y_c)\\
           \pcln~ E(G')=E(G')\cup (x,y)\\
           \pcln ~L=L\cup y_c\\
           \pcln ~\pcreturn y\\
           \pcln \pcendif
         }
        \end{pcvstack}
          \pchspace
        \begin{pcvstack}
          \procedure{ Procedure $\hat{S}_2(R,x)$ \pccomment{Stage II with fixed random string $R$}}{%
          \pcln \pcif \exists (x,y)\in E(G) ~\pcreturn y\\
          \pcln \pcif x=R~\mbox{mark}~ x\\
          \pcln \pcif x \mbox{ is marked} \\
          \pcln \pcind\mbox{Run} \tilde{h}(x)\\
          \pcln \pcind\pcfor \mbox{ every query } x_j \mbox{ made by } \tilde{h}\\
          \pcln \pcind\pcind \mbox{Feed } y=Sim(x_j)\\
          \pcln \pcind\pcendfor\\
          \pcln \pcind\tilde{y}= \tilde{h}(x)\\
          \pcln \pcind\pcif \tilde{y}\neq Sim(x)\\
          \pcln \pcind\pcind \mbox{\sc Bad}=1\\
          \pcln \pcind\pcind \pcreturn \perp\\
          \pcln \pcind\pcendif\\
          \pcln \pcendif\\
          \pcln \pcreturn Sim(x)\\
          }
        \end{pcvstack}
      \end{pchstack}
}}
   \caption{Simulator for Sponge Construction. $S_2$ is initialized with $z,R$ and $\tilde{h}$.}
  \label{fig:simulator-sponge}
\end{figure}

\begin{figure}[!htb]
  \centering
  \fbox{\scalebox{0.85}{
      \begin{pchstack}
        \begin{pcvstack}
          \procedure{Procedure $Newnode$}{%
            \pcln \pcif |L|=2^c \pcreturn \perp\\
           \pcln y_c\sample\bool^{c}\setminus L\\
           \pcln y_r\sample\bool^{r} \\
           \pcln L=L\cup y_c\\
           \pcln \pcreturn y=(y_c,y_r)\\
         }
         \pcvspace
          \procedure{Procedure $Findpath(R,x)$}{%
           \pcln p= \lambda\\
           \pcln \pcif x=r~ \pcreturn p \\
           \pcln \pcelse\\
           \pcln ~~ \mbox{Parse } x= (x_r,x_c)\\
           \pcln ~~\mbox{Find } \hat{x}_r \mbox{ such that} \\
           ~~~\left( (t,(\hat{x}_r,x_c)) \in
             E(G)\right)\wedge t \mbox{ is marked}\\
           \pcln ~~m'= x_r \xor \hat{x}_r\\
           \pcln ~~p'=Findpath(r,t)\\
           \pcln p=p'||m'\\
           \pcln \pcreturn p\\
         }
        \end{pcvstack}
         \pchspace
        \procedure{ Procedure $Sim(x)$ \pccomment{Bertoni etal simulator \cite{EC:BDPV08}}}{%
          \pcln \pcif \exists (x,y)\in E(G) ~\pcreturn y\\
          \pcln L= L\cup \{x_c\}\\
            \pcln \pcif x \mbox{ is marked}\\
                \pcln ~~ m=Findpath(R,x)\\
                \pcln ~~\pcif |m|< \ell\\
                    \pcln ~~~~y=Newnode()\\
                    \pcln~~~~ E(G)=E(G)\cup (x,y)\\
                    \pcln~~~~ \mbox{Mark all node } (*,y_c)\\ 
                    \pcln ~~~~\pcreturn y\\
                \pcln ~~\pcelse \pccomment{$|m|=\ell$} \\
                     \pcln~~ ~~z=\mathcal{F}(m)\\
                     \pcln ~~~~\mbox{Break } z=z_1||z_2||\cdots||z_t
                     \pccomment{$tr=s$}\\
                     \pcln ~~~~ y^{(0)}=x\\
                     \pcln ~~~~\pcfor i=1~\mbox{to}~t \pcdo\\
                     \pcln~~~ ~~~ y_c\sample\bool^{c}\setminus L\\
                        \pcln~~~~~~ y_r=z_i\\
                        \pcln~~~ ~~~ y^{(j)}=(y_r,y_c)\\
                        \pcln~~~~~~ E(G)=E(G)\cup (y^{(j-1)},y^{(j)})\\
                        \pcln ~~~~~~L=L\cup \{y_c\}\\
                     \pcln ~~~~\pcendfor \\
                     \pcln~~~~\pcreturn y^{(1)}\\
                 \pcln ~~\pcendif\\
             \pcln \pcelse \pccomment{x is unmarked}\\
                    \pcln ~~~~y=Newnode()\\
                    \pcln~~~~ E(G)=E(G)\cup (x,y)\\
                    \pcln ~~~~\pcreturn y\\
            \pcln \pcendif       
         }
      \end{pchstack}
}}      
  \caption{Bertoni \etal Simulator for Classical indifferentiability of Sponge Construction.\cite{EC:BDPV08}}
  \label{fig:originalsimulator}
\end{figure}

\subsubsection{Simulating $f$ in Stage II}
\label{sec:simulating-II-stage}
The simulator gets the implementation $\tilde{h}$, along with the advice string $z$. In addition, the simulator receives the random string $R\in \bool^n$.  $S$ initializes by
\emph{marking} the node $R$ in the simulator graph. Following the simulator of \cite{EC:BDPV08}, the idea of marking a
node $x$ is to declare that there is a path in the simulator graph from the root $R$ to $x$.



Now, the simulator invokes $\tilde{h}$ on input $x$. For each query $x_i$ made by $\tilde{h}$, $\hat{S}$ forwards the query to the simulator $Sim$ as a query and upon receiving an answer, forwards it to the distinguisher. Finally when $\tilde{h}(x)$ returns a value, $\hat{S}$ checks whether $\tilde{h}(x)=Sim(x)$. If the check fails, the simulator raises the flag ${\sc Bad}0$ and aborts. Otherwise, it returns $Sim(x)$.

\subsubsection{Proving the Crooked Indifferentiability.}
\label{sec:prov-crook-indiff}
The detail of the games and transitional probabilities are described in Section\ref{sec:game-transitions-sponge}.  The crooked indifferentiability is proved via the following lemma.
\begin{lemma}
  \label{lemma:crookedspongenormal}
  If $\mbox{\sc Bad}$ does not happen then $\epsilon\leq \frac{q^2\tau^2}{2^c}$. Moreover,   \begin{align*}
    \Pr[\mbox{\sc Bad}] \leq q_2(\ell+s)\epsilon\cdot 2^r+ \frac{q_2(\ell+s)\kappa}{2^c}
  \end{align*} 
\end{lemma}

\subsubsection{Proof of Lemma~\ref{lemma:crookedspongenormal}.}
\label{sec:bound-bad-prob}

If $\mbox{\sc Bad}$ does not happen, then our simulator emulates the simulator of \cite{EC:BDPV08} perfectly. Note, the Newnode subroutine does not sample any $x_c$ such that for some $x_r$, $(x_r,x_c) \in z$. Moreover, none of the marked nodes in the tree is subverted. Hence, in that case, the classical indifferentiability simulator perfectly simulates $f$ maintaining consistency with the random oracle $F$. By the classical indifferentiability theorem of \cite{EC:BDPV08},  $\epsilon\leq \frac{q^2\tau^2}{2^c}$. Here $\tau$is the number of queries made by the implementation.

Recall that the event $\mbox{\sc Bad}$ happens if for some marked node $x$ in the graph, $x$ is subverted; $\tilde{h}(x)\neq h(x)$. We say a simulator query to be safe if it is unsubverted or the $c$ part of the output does not match with any input of the transcript of $\adv_1$.
 Observe that for each safe query $x$, $h(x)$ is uniformly distributed.

 \noindent  Let $E_i$ denote the event that $i^{th}$ query made by $\adv$ to the simulator, the input $x_i$ is marked and $x_i$ is not safe. . $x_i=x_r||x_c$. If $x_i$ is marked, then for some marked $x_j$, $j<i$, and for some $m_i\in\bool^r$, $x_i=m||0^c\xor\tilde{h}(x_j)$. Conditioned on $x_j$ is safe, $\tilde{h}(x_j)=h(x_j)$ is uniformly distributed. Hence, $x_c$ is independently distributed. By definition of subversion, and taking union bound over all possible $x_r$, the probability $\tilde{h}(x_i)$ is subverted is bounded by $\epsilon\cdot 2^r$. Given that $\tilde{h}(x_i)$ is not subverted, the probability of the $c$-part of the output matches with the $c$-part of some query of $\adv_1$ is $\frac{\kappa}{2^c}$
 
\begin{align*}
  \Pr [\mbox{\sc Bad}]&\leq \sum_{i=1}^q \Pr [ E_i|\wedge_{j=1}^{i-1} E_j]\\
                       &\leq  \sum_{i=1}^q \left( \epsilon\cdot 2^r\right) +\frac{\kappa}{2^c} \\
  &=q_2(\ell+s)\epsilon\cdot 2^r +\frac{q_2(\ell+s)\kappa}{2^c} 
\end{align*}

\begin{figure}[!htb]
  \centering
  \fbox{\scalebox{0.85}{
      \begin{pchstack}
        \begin{pcvstack}
           Game $(h,C^{\tilde{h}})$ \\\\
           \procedure{$\mathcal{O}_h(x)$ }{%
             \pcln \pcreturn h(x) \\
             }
             \pcvspace
               \procedure{$\mathcal{O}_C(m)$}{
           \pcln x= (x_a,x_c)=(R_0,R_1) \\
            \pcln \pcfor i=0 \mbox{ to } \left\lceil \frac{\ell}{r}\right\rceil-1 \pcdo\\
            \pcln \pcind (x_a,x_c)=\tilde{h}(x_a\xor m_i,x_c)\\
            \pcln \pcendfor\\
            \pcln \pcfor i=0 \mbox{ to } \left\lceil \frac{s}{r}\right\rceil -1 \pcdo\\
            \pcln \pcind \mbox{Append } x_a \mbox{ to output}\\
            \pcln \pcind (x_a,x_c)=\tilde{h}(x_a,x_c)\\
            \pcln \pcendfor\\   
         }
         \end{pcvstack}
           \pchspace
           \begin{pcvstack}
             Game $({\bf G}_0)$ \\\\
           \procedure{$\mathcal{O}_h(x)$ ($i\in [\ell]$)}{%
             \pcln \pcreturn \mathcal{G}(1,x) \\
             }
             \pcvspace
               \procedure{$\mathcal{O}_C(m)$}{
           \pcln x= (x_a,x_c)=(R_0,R_1) \\
            \pcln \pcfor i=0 \mbox{ to } \left\lceil \frac{\ell}{r}\right\rceil-1 \pcdo\\
            \pcln \pcind (x_a,x_c)=\mathcal{G}(0,x_a\xor m_i,x_c)\\
            \pcln \pcendfor\\
            \pcln \pcfor i=0 \mbox{ to } \left\lceil \frac{s}{r}\right\rceil -1 \pcdo\\
            \pcln \pcind \mbox{Append } x_a \mbox{ to output}\\
            \pcln \pcind (x_a,x_c)=\mathcal{G}(0,x_a,x_c)\\
            \pcln \pcendfor\\   
          }
          \pcvspace
          \procedure{$\mathcal{G}(i,x)$}{%
            \pcln y=\tilde{y}=h(x)\\
            \pcln \pcif i=1~~\pcreturn y\\
            \pcln \pcelse\\
             \pcln \pcind \mbox{Run} \tilde{h}(x)\\
          \pcln \pcind\pcfor \mbox{ every query } x_j \mbox{ made by } \tilde{h}\\
          \pcln \pcind\pcind \mbox{Feed } y=h(x_j)\\
          \pcln \pcind\pcendfor\\
          \pcln \pcind\tilde{y}= \tilde{h}(x)\\
          \pcln \pcind\pcreturn \tilde{y}\\
          \pcln \pcendif\\
           } 
         \end{pcvstack}
          \end{pchstack}
    }}
  \caption{Game \pcnotionstyle{Real} and Game ${\bf G}_0$}
  \label{fig:gamereal}
\end{figure}

\subsection{Game Transitions}
\label{sec:game-transitions-sponge}
Our crooked-indifferentiability proof relies on four intermediate game, denoted by ${\bf G}_0,{\bf G}_1,{\bf G}_1,$ and ${\bf G}_2$. We start with the real game $(h, C^{\tilde{h}})$. There are two public interfaces for the adversary to query. The first one is $\mathcal{O}_h$, which can be used to interact with the function $h$. The other one is $\mathcal{O}_C$, which can be used to compute $C^{\tilde{h}}$. In the intermediate games we shall use additional subroutine $\mathcal{G}$.

\noindent \textbf{Game ${\bf G}_0$. } In this game, we introduce the subroutine $\mathcal{G}$ which acts as a common interface to $h,\tilde{h}$. We modify the  $\mathcal{O}_h$ and $\mathcal{O}_C$ subroutine. For queries to $h$, $\mathcal{G}$ is called with parameter $h$ whereas for $\tilde{h}$ the parameter value is set to be 0. These changes are ornamental and do not change the output for any of the query. Hence
\begin{align*}
    \Pr[\adv^{h,C^{\tilde{h}}}=1]=\Pr[\adv^{{\bf G}_0}=1]
\end{align*}

\begin{figure}[!htb]
  \centering
  \fbox{\scalebox{0.85}{
      \begin{pchstack}
          \begin{pcvstack}
             Game ${\bf G}_1a,\mbox{\gamechange{${\bf G}_1b$}}$ \\\\
           \procedure{$\mathcal{O}_h(x)$ }{%
             \pcln \pcreturn \mathcal{G}(1,x) \\
             }
             \pcvspace
               \procedure{$\mathcal{O}_C(m)$}{
           \pcln x= (x_a,x_c)=(R_0,R_1) \\
            \pcln \pcfor i=0 \mbox{ to } \left\lceil \frac{\ell}{r}\right\rceil-1 \pcdo\\
            \pcln \pcind (x_a,x_c)=\mathcal{G}(0,x_a\xor m_i,x_c)\\
            \pcln \pcendfor\\
            \pcln \pcfor i=0 \mbox{ to } \left\lceil \frac{s}{r}\right\rceil -1 \pcdo\\
            \pcln \pcind \mbox{Append } x_a \mbox{ to output}\\
            \pcln \pcind (x_a,x_c)=\mathcal{G}(0,x_a,x_c)\\
            \pcln \pcendfor\\   
          }
          \pcvspace
          \procedure{$\mathcal{G}(i,x)$}{%
             \pcln y=\tilde{y}=h(x)\\
            \pcln \pcif i=1~~\pcreturn y\\
             \pcln \mbox{Run} \tilde{h}(x)\\
          \pcln \pcfor \mbox{ every query } x_j \mbox{ made by } \tilde{h}\\
          \pcln \pcind \mbox{Feed } y=h(x_j)\\
          \pcln \pcendfor\\
          \pcln  \pcif{\tilde{h}(x)\neq h(x)}\\
          \pcln \pcind \mbox{\sc Bad}=1\\
          \pcln  \pcbox{\tilde{y}= \tilde{h}(x)}~~\gamechange{$\tilde{y}= h(x)$}\\
          \pcln \pcreturn \tilde{y}\\
          \pcln \pcendif\\
           } 
         \end{pcvstack}
         \pchspace
          \begin{pcvstack}
             Game $({\bf G}_2)$ \\\\
           \procedure{$\mathcal{O}_h(x)$ }{%
             \pcln \pcreturn \mathcal{G}(1,x) \\
             }
             \pcvspace
               \procedure{$\mathcal{O}_C(m)$}{
           \pcln x= (x_a,x_c)=(R_0,R_1) \\
            \pcln \pcfor i=0 \mbox{ to } \left\lceil \frac{\ell}{r}\right\rceil-1 \pcdo\\
            \pcln \pcind (x_a,x_c)=h(,x_a\xor m_i,x_c)\\
            \pcln \pcendfor\\
            \pcln \pcfor i=0 \mbox{ to } \left\lceil \frac{s}{r}\right\rceil -1 \pcdo\\
            \pcln \pcind \mbox{Append } x_a \mbox{ to output}\\
            \pcln \pcind (x_a,x_c)=h(x_a,x_c)\\
            \pcln \pcendfor\\   
          }
          \pcvspace
          \procedure{$\mathcal{G}(i,x)$}{%
            \pcln y=\tilde{y}=h(x)\\
            \pcln \pcif i=1~~\pcreturn y\\
             \pcln \mbox{Run} \tilde{h}(x)\\
          \pcln \pcfor \mbox{ every query } x_j \mbox{ made by } \tilde{h}\\
          \pcln \pcind \mbox{Feed } y=h(x_j)\\
          \pcln \pcendfor\\
          \pcln  \pcif{\tilde{h}(x)\neq h(x)}\\
          \pcln \pcind \mbox{\sc Bad}=1\\
          \pcln  \tilde{y}= h(x)  \\
          \pcln \pcreturn \tilde{y}\\
          \pcln \pcendif\\
           } 
         \end{pcvstack}
         \end{pchstack}
    }}
  \caption{Game ${\bf G}_1a, {\bf G}_1b,{\bf G}_2$. In the game ${\bf G}_1b$ the highlighted statement will be executed instead of the boxed statement.}
  \label{fig:game12}
\end{figure}

\noindent \textbf{Game ${\bf G}_1a$.} In this game, we modify the subroutine $\mathcal{G}$. When queried with $(0,x)$ for intended value of $\tilde{h}(x)$, $\mathcal{G}$ checks whether $\tilde{h}(x)=h(x)$. If the equality does not hold $\mathcal{G}$ sets the {\sc Bad} flag. However, it still returns $\tilde{h}(x)$. As the output of any query does not change,
\begin{align*}
  \Pr[\adv^{{\bf G}_0}=1]=\Pr[\adv^{{\bf G}_1a}=1]
\end{align*}
\noindent \textbf{Game ${\bf G}_1b$.} In this game, when {\sc Bad} flag is set, $h(x)$ is returned. Everything else remain unchanged. Using the fundamental lemma of game playing proof,
\begin{align*}
 \left| \Pr[\adv^{{\bf G}_1a}=1]-\Pr[\adv^{{\bf G}_1b}=1]\right|\leq \Pr[\mbox{\sc Bad}]
\end{align*}
\noindent \textbf{Game ${\bf G}_2$.}
We note that in game ${\bf G}_1$, all the $\mathcal{G}(x)$ queries made by $O_c$ is answered with $h(x)$. Hence, in Game ${\bf G}_2$, we give $O_c$ direct access to $h$. Output distribution of the game remains exactly same after this change. Hence,
\begin{align*}
  \Pr[\adv^{{\bf G}_1b}=1]=\Pr[\adv^{{\bf G}_2}=1]
\end{align*}
\noindent \textbf{Game ${\bf G}_3$.}
We replace the oracles $(h,C^h)$ by $(S,\mathcal{F}^S)$ where $S$ is the simulator of Bertoni \etal ~\cite{EC:BDPV08}. By the results in \cite{EC:BDPV08},
\begin{align*}
 \left| \Pr[\adv^{{\bf G}_2}=1]-\Pr[\adv^{{\bf G}_3}=1]\right|\leq \frac{q^2\tau^2}{2^c}. 
\end{align*}
 
Note that the numerator in the right hand side is the square of total number of queries made to the simulator, which in Game ${\bf G}_3$, $q\tau$ as the implementation makes $\tau $ many queries for each invocation. 

Finally, we observe that $\mathcal{G}$ works identically with the second stage simulator of Figure\ref{fig:simulator-sponge}. So As we initialize the simulator of \cite{EC:BDPV08} with $L$ from the stage I, the two games are identical. Hence we get, 
\begin{align*}
  \Pr[\adv^{{\bf G}_3}=1]=\Pr[\adv^{(S^{\mathcal{F}},\mathcal{F})}=1]
\end{align*}
Collecting all the probabilities, we get

\begin{align*}
  \Delta_{\adv_2(r, z, R)}\big((h, C^{ \tilde{h} }(R, \cdot))\ ;\ (S^{\mathcal{F}}(H_{z, r}, z, R), \mathcal{F}) \big)&\leq \Pr[\mbox{\sc Bad}] + \frac{q^2\tau^2}{2^c}\\&\leq \frac{q^2\tau^2+q_2(\ell+s)\kappa}{2^c}+ q_2(\ell+s)\epsilon2^r.  
\end{align*}

\begin{figure}[!htb]
  \centering
  \fbox{\scalebox{0.85}{
      \begin{pchstack}
         \begin{pcvstack}
           Game $({\bf G}_3)$ \\\\
           
          \procedure{$\mathcal{O}_h(x)$ }{%
             \pcln \pcreturn \mathcal{G}(1,x) \\
             }
             \pcvspace
               \procedure{$\mathcal{O}_C(m)$}{
          \pcln\pcreturn\mathcal{F}(m)\\
        }
      \end{pcvstack}
      \pchspace
      \begin{pcvstack}
          \procedure{$\mathcal{G}(i,x)$}{%
            \pcln y=\tilde{y}=h(x)\\
            \pcln \pcif i=1~~\pcreturn y\\
             \pcln \mbox{Run} \tilde{h}(x)\\
          \pcln \pcfor \mbox{ every query } x_j \mbox{ made by } \tilde{h}\\
          \pcln \pcind \mbox{Feed } y=Sim(x_j)\\
          \pcln \pcendfor\\
          \pcln  \pcif{\tilde{h}(x)\neq Sim(x)}\\
          \pcln \pcind \mbox{\sc Bad}=1\\
          \pcln  \tilde{y}= Sim(x)  \\
          \pcln \pcreturn \tilde{y}\\
          \pcln \pcendif\\
           } 
         \end{pcvstack}
         \end{pchstack}
    }}
  \caption{Game ${\bf G}_3$}
  \label{fig:game3}
\end{figure}

\section{Conclusion}
\label{sec:conclusion}

In this paper, we revisited the recently introduced crooked indifferentiability notion. We showed that the proof of crooked indifferentiability of enveloped XOR construction in \cite{C:RTYZ18} is incomplete. We developed new technique to prove crooked indifferentiability of the same construction. We also show that the sponge construction with randomized initial value is also crooked indifferentiable secure hash function.


\bibliographystyle{splncs04} 
\bibliography{abbrev3,crypto,sponge}

\begin{thebibliography}{10}
\providecommand{\url}[1]{\texttt{#1}}
\providecommand{\urlprefix}{URL }
\providecommand{\doi}[1]{https://doi.org/#1}

\bibitem{ACNS:AFMV19}
Ateniese, G., Francati, D., Magri, B., Venturi, D.: Public immunization against
  complete subversion without random oracles. In: Deng, R.H.,
  {Gauthier-Uma{\~n}a}, V., Ochoa, M., Yung, M. (eds.) ACNS 19. {LNCS}, vol.
  11464, pp. 465--485. Springer, Heidelberg (Jun 2019).
  \doi{10.1007/978-3-030-21568-2_23}

\bibitem{CCS:AteMagVen15}
Ateniese, G., Magri, B., Venturi, D.: Subversion-resilient signature schemes.
  In: Ray, I., Li, N., Kruegel, C. (eds.) ACM CCS 2015. pp. 364--375. {ACM}
  Press (Oct 2015). \doi{10.1145/2810103.2813635}

\bibitem{EC:BelHoa15}
Bellare, M., Hoang, V.T.: Resisting randomness subversion: Fast deterministic
  and hedged public-key encryption in the standard model. In: Oswald, E.,
  Fischlin, M. (eds.) EUROCRYPT~2015, Part~II. {LNCS}, vol.~9057, pp. 627--656.
  Springer, Heidelberg (Apr 2015). \doi{10.1007/978-3-662-46803-6_21}

\bibitem{C:BelPatRog14}
Bellare, M., Paterson, K.G., Rogaway, P.: Security of symmetric encryption
  against mass surveillance. In: Garay, J.A., Gennaro, R. (eds.) CRYPTO~2014,
  Part~I. {LNCS}, vol.~8616, pp. 1--19. Springer, Heidelberg (Aug 2014).
  \doi{10.1007/978-3-662-44371-2_1}

\bibitem{BertoniDPA07}
Bertoni, G., Daemen, J., Peeters, M., Assche, G.: Sponge functions. ECRYPT Hash
  Workshop 2007  (01 2007)

\bibitem{eurocrypt/BertoniDPA08}
Bertoni, G., Daemen, J., Peeters, M., Assche, G.V.: On the indifferentiability
  of the sponge construction. In: Advances in Cryptology - {EUROCRYPT} 2008,
  27th Annual International Conference on the Theory and Applications of
  Cryptographic Techniques, Istanbul, Turkey, April 13-17, 2008. Proceedings.
  pp. 181--197 (2008), \url{https://doi.org/10.1007/978-3-540-78967-3\_11}

\bibitem{EC:BDPV08}
Bertoni, G., Daemen, J., Peeters, M., {Van Assche}, G.: On the
  indifferentiability of the sponge construction. In: Smart, N.P. (ed.)
  EUROCRYPT~2008. {LNCS}, vol.~4965, pp. 181--197. Springer, Heidelberg (Apr
  2008). \doi{10.1007/978-3-540-78967-3_11}

\bibitem{PKC:CRTYZZ19}
Chow, S.S.M., Russell, A., Tang, Q., Yung, M., Zhao, Y., Zhou, H.S.: Let a
  non-barking watchdog bite: Cliptographic signatures with an offline watchdog.
  In: Lin, D., Sako, K. (eds.) PKC~2019, Part~I. {LNCS}, vol. 11442, pp.
  221--251. Springer, Heidelberg (Apr 2019). \doi{10.1007/978-3-030-17253-4_8}

\bibitem{C:CDMP05}
Coron, J.S., Dodis, Y., Malinaud, C., Puniya, P.: {Merkle-Damg{\aa}rd}
  revisited: How to construct a hash function. In: Shoup, V. (ed.) CRYPTO~2005.
  {LNCS}, vol.~3621, pp. 430--448. Springer, Heidelberg (Aug 2005).
  \doi{10.1007/11535218_26}

\bibitem{FSE:DegFarPoe15}
Degabriele, J.P., Farshim, P., Poettering, B.: A more cautious approach to
  security against mass surveillance. In: Leander, G. (ed.) FSE~2015. {LNCS},
  vol.~9054, pp. 579--598. Springer, Heidelberg (Mar 2015).
  \doi{10.1007/978-3-662-48116-5_28}

\bibitem{C:DPSW16}
Degabriele, J.P., Paterson, K.G., Schuldt, J.C.N., Woodage, J.: Backdoors in
  pseudorandom number generators: Possibility and impossibility results. In:
  Robshaw, M., Katz, J. (eds.) CRYPTO~2016, Part~I. {LNCS}, vol.~9814, pp.
  403--432. Springer, Heidelberg (Aug 2016). \doi{10.1007/978-3-662-53018-4_15}

\bibitem{EC:DGGJR15}
Dodis, Y., Ganesh, C., Golovnev, A., Juels, A., Ristenpart, T.: A formal
  treatment of backdoored pseudorandom generators. In: Oswald, E., Fischlin, M.
  (eds.) EUROCRYPT~2015, Part~I. {LNCS}, vol.~9056, pp. 101--126. Springer,
  Heidelberg (Apr 2015). \doi{10.1007/978-3-662-46800-5_5}

\bibitem{TCC:MauRenHol04}
Maurer, U.M., Renner, R., Holenstein, C.: Indifferentiability, impossibility
  results on reductions, and applications to the random oracle methodology. In:
  Naor, M. (ed.) TCC~2004. {LNCS}, vol.~2951, pp. 21--39. Springer, Heidelberg
  (Feb 2004). \doi{10.1007/978-3-540-24638-1_2}

\bibitem{EC:MirSte15}
Mironov, I., {Stephens-Davidowitz}, N.: Cryptographic reverse firewalls. In:
  Oswald, E., Fischlin, M. (eds.) EUROCRYPT~2015, Part~II. {LNCS}, vol.~9057,
  pp. 657--686. Springer, Heidelberg (Apr 2015).
  \doi{10.1007/978-3-662-46803-6_22}

\bibitem{AC:RTYZ16}
Russell, A., Tang, Q., Yung, M., Zhou, H.S.: Cliptography: Clipping the power
  of kleptographic attacks. In: Cheon, J.H., Takagi, T. (eds.) ASIACRYPT~2016,
  Part~II. {LNCS}, vol. 10032, pp. 34--64. Springer, Heidelberg (Dec 2016).
  \doi{10.1007/978-3-662-53890-6_2}

\bibitem{CCS:RTYZ17}
Russell, A., Tang, Q., Yung, M., Zhou, H.S.: Generic semantic security against
  a kleptographic adversary. In: Thuraisingham, B.M., Evans, D., Malkin, T.,
  Xu, D. (eds.) ACM CCS 2017. pp. 907--922. {ACM} Press (Oct~/~Nov 2017).
  \doi{10.1145/3133956.3133993}

\bibitem{C:RTYZ18}
Russell, A., Tang, Q., Yung, M., Zhou, H.S.: Correcting subverted random
  oracles. In: Shacham, H., Boldyreva, A. (eds.) CRYPTO~2018, Part~II. {LNCS},
  vol. 10992, pp. 241--271. Springer, Heidelberg (Aug 2018).
  \doi{10.1007/978-3-319-96881-0_9}

\bibitem{C:YouYun96}
Young, A., Yung, M.: The dark side of ``black-box'' cryptography, or: Should we
  trust capstone? In: Koblitz, N. (ed.) CRYPTO'96. {LNCS}, vol.~1109, pp.
  89--103. Springer, Heidelberg (Aug 1996). \doi{10.1007/3-540-68697-5_8}

\bibitem{EC:YouYun97}
Young, A., Yung, M.: Kleptography: Using cryptography against cryptography. In:
  Fumy, W. (ed.) EUROCRYPT'97. {LNCS}, vol.~1233, pp. 62--74. Springer,
  Heidelberg (May 1997). \doi{10.1007/3-540-69053-0_6}

\end{thebibliography}

\section*{Appendix}

\section{Leftout Definitions}
\label{sec:leftout-definitions}

\noindent\textsc{Detection Algorithm}. Given an implementation one may check the correctness of the algorithm by comparing the outputs of the implementation with a known correct algorithm. More precisely, we sample $\alpha_1, \ldots, \alpha_t \sample \bin^m$ and then for all $0 \leq i \leq l$, we check whether $\tilde{h}(i, \alpha) \eq h(i, \alpha)$ holds. If it does not hold, the implementation would not be used. It is easy to see that for $\epsilon$-crooked implementation the subversion would not be detected with probability at least $(1- \epsilon)^t$. So for negligible $\epsilon$ this probability would be still close to one for all polynomial function $t$
and so the implementation can be survived for further use. \smallskip

\subsection{Classical Indifferentiability}
IV-based oracle construction $C^{\mathcal{O}}(\cdot, \cdot)$ 
first fixes an initial value $R$ (chosen randomly from an initial value space). Afterwards, on input $M$, it interacts with the oracle $\mathcal{O}$ and finally it returns an output, denoted as $C^{\mathcal{O}}(R, M)$. When the initial value space is singleton (i.e., degenerated), we simply call $C$ an oracle construction. 
{\em An (initial value based) oracle construction $C$ is called $\mathcal{F}$-compatible} if the domains and ranges of $C$ and $\mathcal{F}$ (an ideal primitive) are same. Now we state the definition of indifferentiability of an oracle construction as stated in \cite{C:CDMP05,TCC:MauRenHol04} in our terminologies. 

\begin{definition}[Indifferentiability]
\label{def:indiff} Let  $\mathcal{F}$ be an ideal primitive and $C^P$ be a $\mathcal{F}$-compatible oracle construction. $C$ is said to be {\em $((q_{P}, q_C, q_{\mathrm{sim}}), \varepsilon)$
indifferentiable} from an ideal primitive $\mathcal{F}$ if there exists
a $q_{\mathrm{sim}}$-query algorithm $S^{\mathcal{F}}$ (called simulator) such that for any $(q_P, q_C)$-query algorithm $\adv$, it holds that
$$ \Delta_{\adv}\left((P, C^{P}(\cdot))\ ;\ (S^\mathcal{F}(\cdot), \mathcal{F})\right) < \varepsilon.$$ 
\end{definition}
In the above definition one may include the complexity (time, query etc.) of the adversary and simulator. However, for information theoretic security analysis, we may ignore the time complexity of the simulator as well as the adversary.\footnote{\small One can easily extend the concrete setup to an asymptotic setup. Let $\langle \mathcal{F}_n, P_n \rangle_{n \in \NN}$ be a sequence of primitives and $C(n)$ be a polynomial time $\mathcal{F}_n$-compatible oracle algorithm. $C^{P_n}(n)$ is said to be (computationally) indifferentiable from
	$\mathcal{F}_n$ if there exists a polynomial time simulator $S^{\mathcal{F}_n}$ such that for all polynomial time oracle algorithm $\mathcal{A}$, 
	$\Delta_\adv\left((P_n, C^{P_n}(n))\ ;\ (S^{\mathcal{F}_n}, \mathcal{F}_n)\right) = \negl(n)$.} A popular indifferentiability treatment for hash function considers $\mathcal{F}$ to be a $n$-bit random oracle
which returns independent and uniform $n$-bit strings for every distinct queries. However, the hash function $C^P$ can be defined through different types of primitives $P$ (a random oracle, or a random permutation $\pi_n$, chosen uniformly from the set of all permutations over $\bool^n$). \smallskip

\begin{figure}[h]
\begin{center}
\begin{tikzpicture}[auto, node distance=1cm, >=latex']
\begin{scope}
\tikzstyle{Attacker} = [draw, fill=black!05, rectangle, 
    minimum height=1.7em, minimum width=8em, drop shadow,thick]
\tikzstyle{Box} = [draw, fill=black!05, rectangle, 
    minimum height=1.7em, minimum width= 1.7em, drop shadow,thick]
\tikzstyle{to} = [->,thick]
\tikzstyle{line}= [-,thick]
\tikzstyle{dotto} = [->,dotted, thick]

\node [Box, name=F] {$P$};
\node [Box, right=of F, name=C] {$C$}; 
\node [Box, right=of C, name=S] {$S$}; 
\node [Box, right=of S, name=G] {$\mathcal{F}$}; 
\draw [draw,to] (C) -- (F);
\draw [draw, to] (S) -- (G);

\coordinate (mid) at ($0.5*(C.east) + 0.5*(S.west)$);

\draw [draw, line] (mid) -- ++(0,0.75);
\draw [draw, line] (mid) -- ++(0,-0.75);

\node [Attacker, below of=mid, node distance=2cm, name=D] {$\mathcal{A}$};

\coordinate (dleftmid) at ($0.5*(D.north west) + 0.5*(D.north)$);
\coordinate (drightmid) at ($0.5*(D.north east) + 0.5*(D.north)$);

\draw (dleftmid) to[bend right, thick, dotted] ($0.5*(F.south)+0.5*(dleftmid)$);
\draw ($0.5*(F.south)+0.5*(dleftmid)$) to[->, bend left, thick, dotted] (F.south);
\draw (dleftmid) to[bend left, thick, dotted] ($0.5*(S.south)+0.5*(dleftmid)$);
\draw ($0.5*(S.south)+0.5*(dleftmid)$) to[->, bend right, thick, dotted] (S.south);

\draw (drightmid) to[bend right, thick, dotted] ($0.5*(C.south)+0.5*(drightmid)$);
\draw ($0.5*(C.south)+0.5*(drightmid)$) to[->, bend left, thick, dotted] (C.south);
\draw (drightmid) to[bend left, thick, dotted] ($0.5*(G.south)+0.5*(drightmid)$);
\draw ($0.5*(G.south)+0.5*(drightmid)$) to[->, bend right, thick, dotted] (G.south);

\end{scope}
\end{tikzpicture}
\end{center}
\caption{The distinguishing game of $\adv$ in the indifferentiability security game.} 
\label{fig:indiff}
\end{figure}

\subsection{Markov Inequality}

\begin{lemma}
  Let $X$ be a non-negative random variable and $a>0$ be a real number. Then it holds that
  \begin{align*}
    \Pr[X \geq a] \leq \frac{\ex (X)}{a}
  \end{align*}
\end{lemma}

A simple application of Markov inequality (which is used repeatedly in this paper) is the following. Consider a joint distribution of random variables $X$ and $Y$. Suppose $E$ is an event for which $\Pr((X, Y) \in E) \leq \epsilon$. Let  $f(x) :\eq \Pr((X, Y) \in E | X \eq x)$ and $E_1 :\eq \{x: f(x) \geq \delta\}$. It follows from the definition that $\ex(f(X)) \eq \Pr(E)$. Now, we use Markov's inequality  
\begin{align*}
\Pr(E_1) & \eq \Pr(f(X) \geq \delta) \\
& \leq \ex(f(X))/\delta  \\
 & \eq \epsilon/\delta. 
\end{align*}
Note that when $X$ and $Y$ are independent, $f(x)=\Pr((x, Y) \in E)$




\section{Leftout Proofs}
\label{sec:leftout-proofs}
\subsection{Proof of Lemma~\ref{lemma:Dtilde}}
 For any $\epsilon$-crooked implementation $H$ and every $h \in \textsf{H}|_z$, $0 \leq i \leq l$, we have $\Pr_x(\tilde{h}(i, x) \neq h(i, x)) \leq \epsilon$ where $x \sample \bin^n$. So, for $\alpha \sample \mathcal{D}$, 
	\begin{align*}
	\ex_\alpha(d(\alpha, h)) & \eq  \Pr_\alpha(\tilde{h}(\alpha) \neq h(\alpha) \vee \alpha \in \mathcal{D}(z)) \\
	& \leq  \epsilon \pl q_1 2^{-n}.
	\end{align*}
Now, we fix any $\alpha :\eq (i, x)$ and $1 \leq j \leq \tau$. For any function $g \in \textsf{H}|_z$, let $\mathcal{S}_{\alpha, g} :\eq \{(f,\beta) \in \textsf{H}_z \times \bin^n:\ f|_{ \mathcal{Q}^f_j(\alpha) \rightarrow \beta}=g  \}$. We shall use the following Claim which we prove later. \medskip
	
\begin{Claim}
  \label{claim:sizefb}
For a function $g$, we have $|\{(f,\beta):\ f|_{ \mathcal{Q}^f_j(\alpha) \rightarrow \beta}=g  \}|=N$.
\end{Claim}

\noindent Now for each $j\in \{1,2,\cdots,\ell\}$,
\begin{align*}
	\mathsf{Ex}_{\alpha, h}D^j(\alpha,h)&=\mathsf{Ex}_{\alpha, h}\mathsf{Ex}_\beta \big( d(\alpha, h|_{\tilde{\alpha}_j\rightarrow \beta}) \big) \\ &=\sum_{\alpha,h,\beta}\Pr(h)\Pr(\alpha)\Pr(\beta)\cdot d(\alpha,h|_{\tilde{\alpha}_j\rightarrow\beta})\\
	&= 2^{-n} \sum_{(h, \beta) \in \mathcal{S}_{\alpha, g}} \sum_{\alpha,g}\Pr(g)\Pr(\alpha)\cdot d(\alpha, g)\\
	&= \sum_{\alpha,g}\Pr(g)\Pr(\alpha)\cdot d(\alpha, g)\\	&=\mathsf{Ex}_{\alpha,g}d(\alpha,g)\\
	&\leq \mathsf{Ex}_g (\epsilon \pl q_1 2^{-n}) \\
	& \leq (\epsilon \pl q_1 2^{-n}) \hskip50pt  \qed \medskip
	\end{align*}

\subsection{Proof of Lemma~\ref{lemma:point-function-good}}
\label{sec:proof-lemma-pointfuncition}

\begin{proof}
A pair can be bad in two ways and we bound each case separately. We first bound that for a random $\alpha$ and $h$, either $d(\alpha, h) \eq 1$ or ${D}^j(\alpha, h) > \epsilon_1^{1/2}$. The probability of the first event is clearly bounded by $\epsilon$ and whereas the probability of the second event is bounded by $\epsilon_1^{1/4}$ (Applying Markov inequality and the Lemma \ref{lemma:Dtilde}). So a random $\alpha'$ is queried by some $\alpha$ satisfying the above event holds  with probability at most $\tau(\epsilon \pl \epsilon_1^{1/4}) \leq 2\tau \epsilon^{1/4}$.

By using simple averaging argument, average number of $\alpha'$ for which the number of $\alpha$ that queries $ \alpha'$ is at least $1/\epsilon^{1/4}$ is at most $\tau \epsilon_1^{1/4}$. By adding this two cases we complete our proof.  \qed
\end{proof}
        
\subsection{Proof of Claim~\ref{claim:sizefb}} Suppose $j \eq 1$ and so $\mathcal{Q}^f_j(\alpha) \eq \alpha$. Now, the function $f$ agrees with $g$ except that the output of $f$ at ${\alpha}$ can be any $n$-bit string, which should be $\beta$. Now assume $j > 1$. So, $f(\mathcal{Q}^f_k(\alpha)) \eq g(\mathcal{Q}^f_k(\alpha))$ for all $k < j$. For $k \eq j$, $f(\mathcal{Q}^f_j(\alpha)) \eq \beta$ (for any choice of $\beta \in \bin^n$). For all other inputs, $f$ agrees with $g$. So the claim follows. \qed
\subsection{Proof of Lemma \ref{lemma:goodRh}}
\label{sec:proof-lemma-refl-goodrh}
  From Lemma~\ref{lemma:robusth}, we know that $\Pr_h(h \mbox{ is not robust})\leq \epsilon_2^{1/2} $. Now, fix a robust $h$. Then
\begin{align*}
\Pr_{R}((R, h) \not\in \mathcal{G}^*) & \leq \sum_{m \in \bin^n} \prod_{i \eq 1}^{l} \Pr_{R_i}( ((i, m \oplus R_i), h )\mbox{ is not good})
 \leq 2^n \times \epsilon_2^{l/2}  \leq 1/2^n.
\end{align*}
\noindent Hence, we have
\begin{align*}
  \Pr_{R,h}((R, h) \not\in \mathcal{G}^*) &\leq \Pr_h(h \mbox{ is not good})+  \Pr_{R}((R, h) \not\in \mathcal{G}^*|h \mbox{ is good}) 
                                         \leq \epsilon_2^{1/2}+ 1/2^n.  \qed
\end{align*}

\subsection{Proof of Lemma \ref{lemma:index-independent-beta}}
\label{sec:proof-lemma-refl}

\noindent We start with the following observation.
\begin{observation}
  \label{obs:queryset}
  \begin{align*}
    Q_{\twoheadrightarrow\alpha_i}^h=Q_{\twoheadrightarrow\alpha_i}^{h_\beta}
  \end{align*}
\end{observation}

\begin{proof}
Let $\alpha'\twoheadrightarrow_h\alpha_i$. Suppose $\alpha_i$ is the $j^{th}$ query in the computation of $\tilde{h}(\alpha')$. Clearly,  $j^{th}$ query in the computation of $\tilde{h}(\alpha')$ depends on the query and responses of first $j-1$ queries made by $\tilde{h}(\alpha')$. Let $\tilde{\alpha}_1,\cdots,\tilde{\alpha}_{j-1}$ be those queries. As $\alpha_i\notin \{\tilde{\alpha}_1,\cdots,\tilde{\alpha}_{j-1}\}$, we get $h(\tilde{\alpha}_k)=h_\beta(\tilde{\alpha}_k)$. As all the previous query responses are same, $\alpha'\twoheadrightarrow_{h_\beta}\alpha_i$. 
\qed
\end{proof}

We are now ready to prove the required three statements. Assume that $i$ is the index of interest for $(R,h)$ and $m$. by definition $\alpha_i\notin B_h$. 
\begin{Claim}
  \label{claim:Chtildesame}
  $\alpha_i\notin B_{h_\beta}$ if and only if $\alpha_i\notin B_{h}$.
\end{Claim}

\begin{proof}
  As $\beta\in S$, $d(\alpha_i,h_\beta)=d(\alpha_i,h)$. By Observation \ref{obs:queryset},$ Q_{\twoheadrightarrow\alpha_i}^h=Q_{\twoheadrightarrow\alpha_i}^{h_\beta}$. So, only things that require proofs are the following.  First, for each $k$ it should hold that, $\alpha_i$ is the $k^{th}$ query of $\tilde{h}(\alpha)$ if and only if $\alpha_i$ is also the $k^{th}$ query of $\tilde{h_{\beta}}(\alpha)$. Moreover, it needs to hold that ${D}^k(\alpha, h)={D}^k(\alpha, h_\beta)$. The first statement follows from the following observation. As all the previous $k-1$ query responses are same in both cases, the $k^{th}$ query made by the implementation on $\alpha$ will be same. Hence $\alpha_i$ is the $k^{th}$ query of $\tilde{h}(\alpha)$, if and only if $\alpha_i$ is also the $k^{th}$ query of $\tilde{h_{\beta}}(\alpha)$. The second statement follows from the definition.
  \begin{align*}
    {D}^k(\alpha, h)=\mathsf{Ex}_{\beta'}(d(\alpha,h_{\alpha_i\rightarrow \beta'}))={D}^k(\alpha, h_\beta)
  \end{align*}
  \qed
\end{proof}

The next claim proves the second statement
\begin{Claim}
  For all $j<i$, $\alpha_j\not\twoheadrightarrow_{h_\beta} \alpha_i$ where $\alpha_j=(j,m\xor R_j)$.
\end{Claim}

\begin{proof}
  Fix a $j<i$. We know $\alpha_j\not\twoheadrightarrow_{h} \alpha_i$. This implies $\alpha_j\notin  Q_{\rightarrow\alpha_i}^h$. By Observation \ref{obs:queryset}, $\alpha_j\notin  Q_{\rightarrow\alpha_i}^{h_\beta}$. Hence, $\alpha_j\not\twoheadrightarrow_{h_\beta} \alpha_i$. \qed
\end{proof}

The final step is to prove that $i$ is the minimum such index  even for $(R,h_\beta)$ that satisfies first two points. For that we prove the following,

\begin{Claim}
  If there exists and index $i'<i$ such that $\alpha_{i'}\notin B_{h_\beta}$ and $\forall~j<i' \alpha_j\not\twoheadrightarrow_{h_\beta}\alpha_{i'}$, then it holds that $\alpha_{i'}\notin B_{h}$ and $\forall~j<i'$,$ \alpha_j\not\twoheadrightarrow_{h}\alpha_{i'}$
\end{Claim}

\begin{proof}
  Using Claim \ref{claim:Chtildesame}, we get
  \begin{align*}
    \alpha_{i'}\notin B_{h_\beta}\implies\alpha_{i'}\notin B_{h}
  \end{align*}

  For the second part, we show the contrapositive. Suppose for some $j<i'$, $\alpha_j\twoheadrightarrow_{h}\alpha_{i'}$. Hence $\alpha_j\in Q^h_{\rightarrow \alpha_{i'}}$. Moreover from assumption $\alpha_j\not\twoheadrightarrow_{h}\alpha_{i}$. All the queries made by $\tilde{h}(\alpha_j)$ and the responses remain same in $\tilde{h_\beta}(\alpha_j)$. This implies $\alpha_j\twoheadrightarrow_{h_\beta}\alpha_{i'}$.
\qed  
\end{proof}

\subsection{Proof of Proposition \ref{prop:key}}
\label{sec:proof-prop-key}

\begin{proof}
  We let $B_j$ denote the event of the choice of $h(\alpha_i)$ such that $(0, \tilde{g}_R^{h_\beta}(m) )\in \mathcal{D}(L_0) \cup C^{h_\beta} \wedge (R, h_\beta, F) \vdash \tau$ holds.

  \begin{align*}
    |B_j| \leq \epsilon 2^n \pl (q_1 \pl i).
  \end{align*}
  
  Fix a function $h$. $E_j=\tau_j\wedge i\mbox{ is the index of interest for }(m,R,h)\wedge \{h_{\alpha_i\rightarrow\beta}\}_{\beta\in S}\}$.
  As there are $|S|$ many choices of $h(\alpha_i)$ in $E_j$, for all $y\in S$
  \begin{align*}
    \Pr[h(\alpha_i)=y|E_j]\leq \frac{1}{|S|}
  \end{align*}

  Hence

  \begin{align*}
    \Pr[B_j|E_j]\leq \frac{|B_j|}{|S|}\leq \frac{2|B_j|}{2^n} \leq 2\epsilon \pl \frac{2(q_1 \pl i)}{2^n}
  \end{align*}

In the second inequality we use that $|S|\geq (1-\epsilon_1^{1/2})2^n> 2^{n-1}$.
Taking sum over all candidate $h$ and thus all choice of $i$, we get the proposition.
\qed
\end{proof}

\end{document}